\documentclass[draftclsnofoot, onecolumn, comsoc, 12pt]{IEEEtran} 

\usepackage{graphicx,epsfig}
\usepackage[noadjust]{cite}
\usepackage{mcite}
\usepackage{amsfonts,helvet}
\usepackage{fancyhdr}
\usepackage{threeparttable}
\usepackage{epsf,epsfig}
\usepackage{amsthm}
\usepackage{amsmath}
\usepackage{siunitx}
\usepackage{amssymb}

\usepackage{dsfont}
\usepackage{subfigure}
\usepackage{color}
\usepackage[linesnumbered,ruled,noend]{algorithm2e}
\usepackage{algpseudocode}
\usepackage{algcompatible}
\usepackage{enumerate}
\usepackage{gensymb}
\usepackage{cancel}

\newcommand{\fig}[1]{Fig.\ \ref{#1}}

\newtheorem{corollary}{Corollary}

\newtheorem{lemma}{Lemma}

\newtheorem{remark}{Remark}

\usepackage{eucal}

\newcommand{\bb}[1]{\mathbb{#1}}
\newcommand{\bF}[1]{\mathbf{#1}}
\newcommand{\Cal}[1]{\CMcal{#1}}

\newcommand\blfootnote[1]{%
  \begingroup
  \renewcommand\thefootnote{}\footnote{#1}%
  \addtocounter{footnote}{-1}%
  \endgroup
}

\begin{document}

\title{\huge Secure Transmission for Hierarchical Information Accessibility in Downlink MU-MIMO}

\author{Kanguk~Lee, Jinseok~Choi, Dong~Ku~Kim, and Jeonghun~Park
\thanks{K. Lee and J. Park are with the School of Electronics Engineering, College of IT Engineering, Kyungpook National University, Daegu, 41566, South Korea (e-mail: {\texttt{kanguk.lee@knu.ac.kr} and \texttt{jeonghun.park@knu.ac.kr}}). J. Choi is with Department of Electrical Engineering, Ulsan National Institute of Science and Technology, South Korea (e-mail: {\texttt{jinseokchoi@unist.ac.kr}}). D. Kim is with Department of Electrical and Electronic Engineering, Yonsei University, South Korea (e-mail: {\texttt{dkkim@yonsei.ac.kr}})
}
 \thanks{This work was supported by the National Research Foundation of Korea (NRF) grant funded by the Korea government (MSIT) (No. 2019R1G1A1094703) and (No. 2021R1C1C1004438).}

}
 
\maketitle \setcounter{page}{1} 

\begin{abstract}
Physical layer security is a useful tool to prevent confidential information from wiretapping. In this paper, we consider a generalized model of conventional physical layer security, referred to as hierarchical information accessibility (HIA). A main feature of the HIA model is that a network has a hierarchy in information accessibility, wherein decoding feasibility is determined by a priority of users. Under this HIA model, we formulate a sum secrecy rate maximization problem with regard to precoding vectors. This problem is challenging since multiple non-smooth functions are involved into the secrecy rate to fulfill the HIA conditions and also the problem is non-convex. To address the challenges, we approximate the minimum function by using the LogSumExp technique, thereafter obtain the first-order optimality condition. One key observation is that the derived condition is cast as a functional eigenvalue problem, where the eigenvalue is equivalent to the approximated objective function of the formulated problem. Accordingly, we show that finding a principal eigenvector is equivalent to finding a local optimal solution. To this end, we develop a novel method called generalized power iteration for HIA (GPI-HIA). Simulations demonstrate that the GPI-HIA significantly outperforms other baseline methods in terms of the secrecy rate. 
\end{abstract}

\section{Introduction}

\blfootnote{A part of this paper will be presented in IEEE Globecom 2021 \cite{lee:gc:21}.} 
As the amount of information delivered through a wireless medium rapidly increases, security in wireless communications becomes a critical issue to prevent leakage of confidential information. Due to the broadcast nature of a wireless medium, it is challenging to implement a secure communication system. Considering a cellular network, for example, if an eavesdropper is located within a coverage region, it is not possible to physically prevent the eavesdropper from overhearing the signals transmitted from a base station (BS). Classical approaches to protect the information from eavesdropping rely on cryptography \cite{massey:proc:88}. Unfortunately, it requires high implementation costs caused by the key distribution and complicated encryption algorithms. As a complement, physical layer security \cite{wyner:bell:75} has been gaining attention. 


From an information theoretical point-of-view, physical layer security shows that a transmitter can reliably send a confidential message to legitimate users with a positive rate, ensuring that an eavesdropper is not able to decode it provided that the eavesdroppers' channel quality is no better than that of the legitimate users' channel. 
This non-zero transmission rate is referred to as {\it{the secrecy rate}}. 
One underlying assumption of physical layer security so far is that two classes of receivers exist in a network: legitimate users and malignant eavesdroppers. 
Nevertheless, this traditional binary security configuration is limited in that it is infeasible to reflect a complicated security structure that will be used in 6G. 
Specifically, as applications of wireless communications have become diversified, a network can have a hierarchy in information access \cite{porambage:ojcomm:21}, so that even some legitimate users can be prohibited to decode particular messages depending on their security clearances.

Motivated by this, \cite{zhang:tvt:19} presented a new model that generalizes physical layer security by incorporating a hierarchical security structure. In this model, each user is assigned into a layer whose priority is different, and the decoding feasibility is determined by the layer. For example, users in a higher security level layer are permitted to decode a message intended to lower security level layers, while the opposite direction access is prohibited. We refer this model as hierarchical information accessibility (HIA). 
The HIA model is useful since it is a generalization of  a conventional physical layer security model.
Specifically, assuming that only two layers exist and a transmitter only delivers a message to one of the two layers, the corresponding HIA model is reduced to a conventional physical layer security setup that assumes legitimate users and illegal eavesdroppers. 
An efficient secure transmission method for the HIA, however, has not been known yet. 
Although \cite{zhang:tvt:19} investigated a transmit power minimization problem, a general secure transmission solution to maximize the sum secrecy rate for the HIA is missing. In this paper, we aim to fill this missing block by proposing a novel method.

\subsection{Prior Works}
There have been several prior works that developed secure precoding solutions to maximize the secrecy rate. Focusing on a multi-user setup considered in this paper, a strategy to exploit the multi-user interference in a beneficial way to degrade eavesdropper's channel quality was developed in \cite{li:tvt:18}. In \cite{zhao:tifs:15}, a sum secrecy rate optimization problem was formulated and a successive convex approximation technique was presented to relax the problem. 
In \cite{sheng:tsp:18}, assuming a wireless network that consists of multiple transmitter-receiver pairs and one eavesdropper, algorithms to maximize the secrecy rate and the secrecy energy efficiency were presented. 
Extending \cite{sheng:tsp:18}, considering multiple eavesdroppers, \cite{choi:wclett:20} proposed a secure transmission algorithm to maximize the sum secrecy rate. 
In \cite{choi:tvt:21}, when multiple eavesdroppers collude to decode confidential messages, an optimization framework to maximize the sum secrecy rate was proposed. 
In \cite{yang:tcom:13, yang:tcom:16}, secure antenna selection methods were investigated. 




In common, the aforementioned prior works considered conventional physical layer security, wherein two classes of receivers exist. 
Assuming the generalized HIA model, in \cite{zhang:tvt:19}, a precoding design to minimize the transmit power was developed under quality of service constraints. In \cite{ding:tcom:17}, a non-orthogonal multiple access (NOMA) scheme with multicast-unicast messages was considered. In this scenario, a power allocation with successive interference cancellation (SIC) was developed to enhance the security performance. Similar to this, a secure transmission with NOMA was also studied in \cite{li:tvt:17}. What is missing in the literature is a general precoding solution that maximizes the sum secrecy rate in the HIA; yet such a solution is necessary to reap de facto performance gains from the HIA model. 
Finding such a solution, however, is particularly difficult since a sum secrecy rate maximization problem is non-convex, where finding a global optimum solution is infeasible. Even worse, to guarantee the HIA condition, the multiple minimum (or maximum) functions are complicatedly intertwined into the secrecy information rate, which makes the problem harder to solve.



\subsection{Contributions}

{\color{black}In this paper, we put forth a secure precoding method to maximize the sum secrecy rate of HIA systems.}
We consider a single-cell downlink system, where a multiple-antenna base station (BS) serves multiple single-antenna users. In this system, the HIA is considered, wherein the users are assigned to a particular layer whose security level is different. It is required to ensure that the users in the higher layer are able to decode the messages intended to the lower layers, while the users in the lower layer cannot decode messages intended to the higher layers. 
Assuming $K$ layers, the BS sends $K$ independent messages, where the same message is intended to the users in the same layer, i.e., multi-group multicast message scenario \cite{hsu:tvt:17, sadeghi:twc:18}.
{\color{black} In such a setup, our main contributions are summarized as follows:}
\begin{itemize}
    
    \item  To accomplish the HIA condition, we adopt a notion of physical layer security. The secrecy rate of the message for layer $k$ is determined by the minimum value of the rates that can be achieved at the users in the higher layers ($\ge k$, to guarantee that the higher layer users can decode it), subtracted by the wiretapping lower layer users' rates ($< k$, to guarantee that the lower layer users cannot decode it). The wiretapping lower layer users' rate is determined differently depending on whether the lower layer users collude or not. Assuming the non-colluding case, the wiretapping channel's rate is determined by the maximum value of the rates that can be achieved at users in the lower layers. In the colluding case, the effective SINR of the wiretapping channel's rate is the sum of the SINR of each lower layer user.  
    
    \item Considering each non-colluding and colluding case, we characterize the achievable secrecy rates. Leveraging this, we formulate optimization problems to maximize the sum secrecy rate with regard to precoding vectors. Unfortunately, the formulated problems are challenging to solve since the multiple non-smooth minimum and maximum functions are complicatedly involved into the secrecy rate and the problems are non-convex. To resolve this, we first approximate our problem using the LogSumExp technique that makes the problem smooth. Then, we reformulate the problem as a form of Rayleigh quotients by rewriting the optimization variables onto a higher dimensional vector. With this form, we derive the first-order optimality condition and show that the derived optimality condition is cast as a functional eigenvalue problem. One remarkable point is that the corresponding matrix is a function of the eigenvector itself, and the eigenvalue is equivalent with the approximated objective function. Accordingly, finding the principal eigenvector is equivalent to finding the local optimal point that has zero gradient. 
    Based on this insight, we propose an algorithm inspired by power iteration, referred as generalized power iteration for HIA (GPI-HIA) that finds the principal eigenvector of the derived functional eigenvalue problems. 
    
    \item In the HIA model, due to the hierarchy of the information access, sequential decoding with SIC is a natural choice at the high layer users, which makes similarity between HIA and downlink NOMA systems. We show that the main problem to maximize the sum secrecy rate in the HIA model can be reduced to a sum rate maximization problem in typical downlink NOMA systems with a fixed decoding order. Specifically, by assuming that only one user is included in each layer and ignoring security to protect the higher layer message from the lower layer users, our HIA model becomes equivalent to downlink multi-antenna NOMA. Leveraging this, we also propose a sum rate maximization precoding method for downlink NOMA with a fixed decoding order.  
    
    \item In numerical results, we validate the performance of the proposed GPI-HIA. Thanks to the fact that the GPI-HIA properly incorporates the complicated rate conditions into its optimization process, it achieves more than $400\%$ secrecy rate gains over an existing convex relaxation based precoding method. 
    Further, we also show that a fairness issue caused by an imbalance between the layers is resolved by modifying the proposed method. 
    Additionally, we present that the proposed precoding method for downlink NOMA also provides considerable rate gains. In addition to those performance benefits, the proposed method does not require any off-the-shelf optimization solver such as CVX.
    In this sense, our method is beneficial not only in a performance perspective, but also in an implementation perspective.

\end{itemize}

\textit{Notation}:
the superscripts $(\cdot)^{\sf T}$, $(\cdot)^{\sf H}$, and $(\cdot)^{-1}$ denote the transpose, Hermitian, and matrix inversion, respectively. $\| \cdot \|_{\sf F}$ and $\| \cdot \|$ denote the Frobenious norm and the $\ell_2$ norm. 
${\bf{I}}_N$ is the identity matrix with size $N \times N$. Assuming that ${\bf{A}}_1, ..., {\bf{A}}_K \in \mathbb{C}^{N \times N}$, ${\bf{A}} = {\rm blkdiag}\left({\bf{A}}_1, ...,{\bf{A}}_k,..., {\bf{A}}_K \right) \in \mathbb{C}^{NK \times NK}$ is a block diagonal matrix.

\section{System Model} \label{sec:sys_model}

\subsection{Hierarchical Information Accessibility}
We consider a single-cell downlink MIMO network, where a BS equipped with $N$ antennas serves $M$ single antenna users. 
In our HIA model, there exist $K$ layers, where each layer includes a subset of the users. Denoting the user set as $\CMcal{M} = \{1,\cdots, M\}$ and the $k$-th layer as $\CMcal{L}_k$ for $k\in\{1,\cdots,K\}$, all the users are assigned to a specific layer, i.e., $\CMcal{M} = \bigcup_{k=1}^K \CMcal{L}_k$ and each user is not allocated to more than one layer,
i.e., $\CMcal{L}_i \cap \CMcal{L}_j = \emptyset, \; i\neq j$. The users in the same layer receive the same message, which corresponds to a multi-group multicast message scenario \cite{hsu:tvt:17,sadeghi:twc:18}. For instance, the BS transmits the message $s_k$ to the users in $\CMcal{L}_k$.

In the considered HIA model, the users assigned to different layers have different security priorities that determine information accessibility. This priority is indicated by the index of the layer. For example, if $i > j$, the users in $\CMcal{L}_i$ have higher priority than the users in $\CMcal{L}_j$. In the HIA, the users with the higher priority are able to access the lower priority information, i.e., the users in $\CMcal{L}_i$ can decode the message $s_j$ for $i\geq j$. On the contrary, the users with the lower priority are prohibited to access the information intended the higher priority users. That is to say, the users in $\CMcal{L}_j$ should not decode the message $s_i$ for $i > j$. From a viewpoint of the physical layer security, the users assigned to the lower priority layer are treated as eavesdroppers to the users in the higher priority layer.

We illustrate an example of the considered HIA system model in \fig{fig:system model}. As observed in the figure, the users in the higher priority layer, i.e., $\CMcal{L}_3$, are permitted to decode the message delivered to the lower priority layers, i.e., $\CMcal{L}_1$ and $\CMcal{L}_2$. On the contrary to that, the users in the lower priority layers are prohibited to decode the  message with  higher priority.

\begin{figure}[!t]
     \centerline{\resizebox{0.45\columnwidth}{!}{\includegraphics{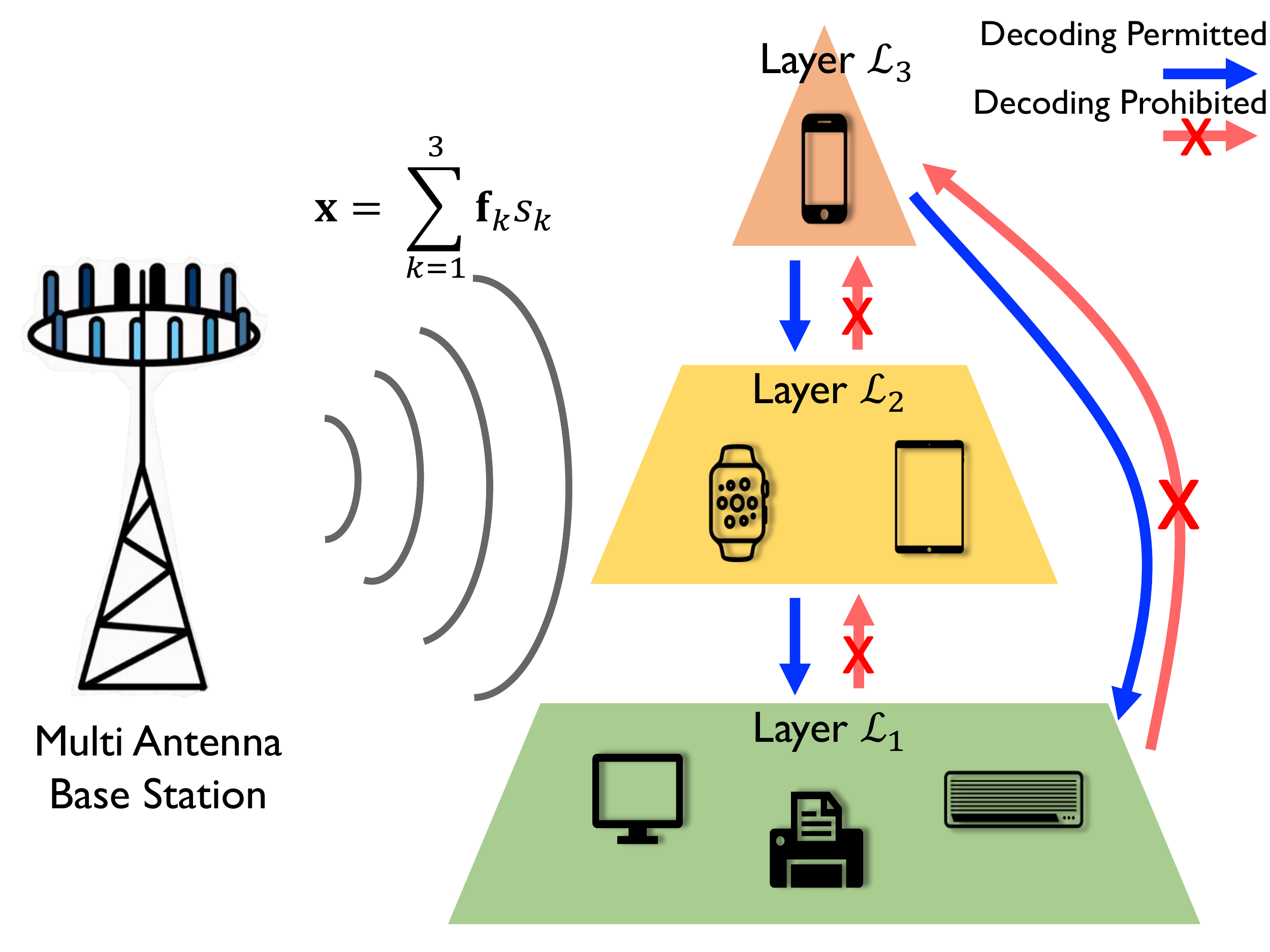}}}
     \caption{Generalized illustration of the considered HIA when $K=3$, $|\CMcal{L}_1| = 3$, $|\CMcal{L}_2| = 2$ and $|\CMcal{L}_3| = 1$. }
     \label{fig:system model}
     \vspace{-1 em}
\end{figure}

\subsection{Channel Model}
We denote the channel vector from the BS to the $i$-th user included in $\CMcal{L}_\ell$ as $\bF{h}_{i,\ell}\in \bb{C}^{N\times 1}$ for $i\in\CMcal{L}_\ell$. 
This channel vector follows the correlated Gaussian distribution, i.e., $\bF{h}_{i,\ell}\sim \mathcal{CN}(\bF{0},\bF{R}_{i,\ell})$ where $\bF{R}_{i,\ell} = \bb{E}\left[\bF{h}_{i,\ell}\bF{h}_{i,\ell}^{\sf H}\right]\in \bb{C}^{N\times N}$ is the channel covariance matrix. The channel covariance matrix is constructed according to the geometric one-ring scattering model \cite{adhikary:tit:13}. Specifically, assuming that the BS is equipped with uniform circular array of $N$ isotropic antennas and radius $\lambda D$ where $\lambda$ is wavelength and $D = \frac{0.5}{\sqrt{(1-\cos(2\pi/N))^2+\sin(2\pi/N)^2}}$, the channel correlation coefficient between $n$-th antenna and $m$-th antenna of ${\bf{R}}_{i, \ell}$ is obtained as 
 \begin{equation}\label{eq:correlation coefficient}
\left[\bF{R}_{i,\ell}\right]_{n,m} = \frac{\beta_{i,\ell}}{2\Delta_{i,\ell}}\int_{\theta_{i,\ell}-\Delta_{i,\ell}}^{\theta_{i,\ell}+\Delta_{i,\ell}}e^{-j\frac{2\pi}{\lambda}\Psi(x)(\bF{r}_n-\bF{r}_m)}\text{d}x,
\end{equation}
where $\beta_{i,\ell}$ is a large scale fading, $\Delta_{i,\ell}$ is a angular spread, $\theta_{i,\ell}\in [0,2\pi]$ is an angle-of-arrival (AoA), $\Psi(x) = [\cos(x),\sin(x)]$ is a wave vector for a planar wave colliding with the angular $x$, and $\bF{r}_n$ is a position vector of $n$-th antenna of BS. By employing the Karhunen-Loeve model as in \cite{adhikary:tit:13},  \cite{dai:twc:16}, the channel vector $\bF{h}_{i,\ell}$ is characterized as $\bF{h}_{i,\ell} = \bF{U}_{i,\ell}\bF{\Lambda}^{\frac{1}{2}}_{i,\ell}\bF{g}_{i,\ell}$, where $\bF{U}_{i,\ell}$ contains the eigenvectors of $\bF{R}_{i,\ell}$, $\bF{\Lambda}_{i,\ell}$ is a diagonal matrix whose elements are non-zero eigenvalues of $\bF{R}_{i,\ell}$, and $\bF{g}_{i,\ell}$ is an independent and identically distributed channel vector drawn from $\mathcal{CN}({\bf 0},{\bf{I}})$. Further, we consider a block fading channel, so that $\bF{g}_{i,\ell}$ is invariant during a channel coherence time. 
In this paper, we assume that the BS knows the perfect CSIT. 

\subsection{Signal Model}
As the multi-group multicast scenario is assumed, $K$ independent messages are transmitted from the BS \cite{hsu:tvt:17, sadeghi:twc:18}. Using linear precoding, the transmit signal ${\bf x}\in \bb{C}^{N\times 1}$ is represented as
\begin{equation}\label{eq:transmit signal}
    \bF{x} = \sum_{k=1}^K\bF{f}_ks_k,
\end{equation}
with the power constraint $\sum_{k=1}^K\|\bF{f}_k\|^2 = 1$ where $\bF{f}_k$ is a $N\times 1$ precoding vector for $s_k$.

Now we explain the decoding process in the HIA model. Recalling that the users in the higher priority can decode the messages sent to the lower priority layers, we assume that the decoding at each user proceeds sequentially from messages sent to the lowest layer to the highest layer, i.e., $1 \rightarrow K$. 
With this sequential decoding architecture, SIC is inherently exploited to eliminate the lower layer messages after decoding them. 
Using SIC, when the user in $\CMcal{L}_\ell$ attempts to decode $s_k$ where $k \le \ell$, the interference only comes from the messages intended to the layer $\CMcal{L}_{k+1}, \cdots, \CMcal{L}_{K}$. For example, letting $K = 3$, the user in $\CMcal{L}_3$ has an access to the messages $s_1$, $s_2$, and $s_3$. This user first attempts to decode $s_1$, while treating $s_2$ and $s_3$ as interference. Once the user succeeds to decode $s_1$, the user eliminates $s_1$ from the received signal and decodes $s_2$ with the reduced amount of the interference. After cancelling $s_2$, the user finally decodes $s_3$. 

Considering the SIC decoding process for the HIA model, the signal model of the user $i$ in $\CMcal{L}_\ell$ when decoding $s_k$ is written as 
 \begin{align}\label{eq:received signal}
     y_{k,i,\ell} = \underbrace{\bF{h}_{i,\ell}^{\sf H}\bF{f}_k s_k}_{\text{desired signal}}+\underbrace{\sum_{j = k+1}^K\bF{h}_{i,\ell}^{\sf H}\bF{f}_j s_j}_{\text{interference}} + n_{i,\ell},
\end{align}
where $n_{i,k}\sim \mathcal{CN}(0,\sigma^2)$ is additive white Gaussian noise. As explained above, we observe that the interference only comes from the higher layer messages, i.e., $k+1, \cdots, K$. Assuming that $s_k$ is drawn from a Gaussian distribution, i.e., $\mathcal{CN}(0,P)$ where $P$ is symbol power, the achievable rate of the message $s_k$ at the user $i$ in $\CMcal{L}_\ell$ is determined as 
\begin{equation}\label{eq : achievable rate for message}
	R_{k,i,\ell} = \log_2\left(1+\frac{\left| \bF{h}_{i,\ell}^{\sf H}\bF{f}_k\right|^2}{\sum_{j = k+1}^K\left|\bF{h}_{i,\ell}^{\sf H}\bF{f}_j\right|^2 + \frac{\sigma^2}{P}}\right).
\end{equation}

\subsection{Performance Metrics and Problem Formulation}

In this subsection, we characterize the secrecy rate for the HIA model and formulate main problems. To prevent the low priority users from decoding the high priority messages, we adopt a notion of physical layer security. From a perspective of the high priority message, the low priority users behave as eavesdroppers. Using physical layer security, by appropriately determining the secrecy rate as a function of the SINR of the low priority users, secure transmission is accomplished. The secrecy rate is determined in a different way depending on whether the low priority users collude or not. In the following, we characterize the secrecy rate considering the non-colluding and the colluding case respectively.

To simplify the secrecy rate performance characterization, we assume the worst case \cite{zhou:jsac:18, tian:splett:17}, where the eavesdroppers (low priority users) in $\CMcal{L}_{\ell'}$ already remove the messages $s_j$ when attempting to decode the message $s_k$, $\ell'<j<k$. 
We note that this assumption results in a lower bound performance as it increases the eavesdropper's SINR. 




\subsubsection{Non-Colluding Case}
first, we consider the non-colluding case. In this case, we assume that the low priority users do not cooperate to decode the high priority message. 
Then, the secrecy rate for $s_k$ is obtained as 
\begin{equation}\label{eq:information rate non colluding}
    \bar{R}_k^{\sf nc} = \left[\min_{i\in \CMcal{L}_\ell, \ell\geq k}R_{k,i,\ell} - \max_{i'\in \CMcal{L}_{\ell'}, \ell' < k}R_{k,i',\ell'} \right]^+.
\end{equation}    
The first term at the right-hand side (r.h.s.) of  \eqref{eq:information rate non colluding} is the minimum rate among the users allowed to decode $s_k$. This is for ensuring that all the users in the layers $\ell \ge k$ are able to decode $s_k$. The second term at the r.h.s. of \eqref{eq:information rate non colluding} is the maximum rate among the wiretapping users, i.e., the users in the layers that has no right to decode $s_k$. 
This is because, in the presence of multiple non-colluding eavesdroppers, no eavesdropper can decode a message when the strongest eavesdropper cannot decode it \cite{kampeas:isit:16, yang:tcom:16, choi:wclett:20}. For this reason, the required redundancy to protect $s_k$ is determined by the maximum rate among all the eavesdroppers, i.e., the users in the layers $\ell' < k$. 
Consequently, in the non-colluding case, the HIA condition is fulfilled when the information rate of $s_k$ is set as $\bar R_k^{\sf nc}$ in \eqref{eq:information rate non colluding}.

Now, we aim to maximize the sum secrecy rate with respect to the precoding vectors. Accordingly, the optimization in the non-colluding case is formulated as
\begin{align}
\max_{\bF{f}_1,\cdots,\bF{f}_K}&\sum\limits_{k=1}^K \bar{R}_k^{\sf nc} \label{eq : objective maximize align non coll}
\\
 \text{subject to }&\sum\limits_{k=1}^{K}||\bF{f}_{k}||^2  = 1.\label{eq : const non coll}
\end{align}
Unfortunately, it is infeasible to directly solve \eqref{eq : objective maximize align non coll} because of two main difficulties: {\it{i)}} \eqref{eq : objective maximize align non coll} is non-convex, wherein finding a global optimal solution is infeasible, {\it{ii)}} for the HIA condition, the rate of each message is determined in a complicatedly intertwined way, including the multiple minimum functions.

\subsubsection{Colluding Case}
next, we consider the colluding case; we assume that the low priority users cooperate to jointly decode the high priority messages. Then, the wiretapping channels form a virtual multiple receive antenna channel, so that the effective SINR is the sum of individual SINR of each eavesdropper \cite{choi:tvt:21}. Accordingly, the secrecy rate for $s_k$ is characterized as 
\begin{align} \label{eq:information rate colluding}
    \bar{R}_k^{\sf c} = \left[\min_{i\in \Cal{L}_\ell, \ell\geq k}R_{k,i,\ell} - R_{k,i',\ell'}^{\sf e} \right]^+,
\end{align}
where 
\begin{align}
    R_{k,i',\ell'}^{\sf e}=&\log_2\left(1+\sum_{\ell' = 1}^{k-1}\sum_{i'\in \Cal{L}_{\ell'}}\frac{\left| \bF{h}_{i',\ell'}^{\sf H}\bF{f}_k\right|^2}{\sum_{j = k+1}^{K}\left|\bF{h}_{i',\ell'}^{\sf H}\bF{f}_j\right|^2 + \frac{\sigma^2}{P}}\right).
\end{align}
We note that the difference of \eqref{eq:information rate colluding} from \eqref{eq:information rate non colluding} is the second term at the r.h.s., i.e., the achievable rate of the wiretapping channel. 
In the colluding case, the HIA condition is satisfied provided that the information rate of $s_k$ is set as $\bar R_k^{\sf c}$ in \eqref{eq:information rate colluding}.

As in the non-colluding case, we aim to maximize the sum secrecy rate with respect to the precoding vectors by solving the following problem:
\begin{align}
	\max_{\bF{f}_1,\cdots,\bF{f}_K}&\sum\limits_{k=1}^K \bar{R}_k^{\sf c} \label{eq : objective maximize align coll}
	\\
 	\text{subject to }&\sum\limits_{k=1}^{K}||\bF{f}_{k}||^2  = 1.\label{eq : const coll}
\end{align}
Similar to the non-colluding case, it is infeasible to directly solve \eqref{eq : objective maximize align coll} due to its non-convexity and non-smoothness.

In the next sections, we put forth the proposed methods to find a local optimal solution of \eqref{eq : objective maximize align non coll} and \eqref{eq : objective maximize align coll} by resolving the challenges.

\begin{remark} [HIA as a generalized model] \normalfont
The presented HIA model generalizes conventional physical layer security. For instance, assuming that $K = 2$ and ${\bf{f}}_1 = {\bf{0}}$, the users included in $\CMcal{L}_1$ are eavesdroppers attempting to overhear $s_2$ and the users in $\CMcal{L}_2$ are legitimate users receiving multicast message $s_2$. This is equivalent to conventional physical layer security with a multicast message setup \cite{li:icc:11}. 

In addition to that, the HIA model also can be reduced to a downlink NOMA scenario. For instance, we assume that a single-user is included in each layer, i.e., $|\CMcal{L}_k| = 1$ and ignore security, i.e., we do not care lower layer users to overhear higher layer messages. Then, the optimization problem \eqref{eq : objective maximize align non coll} reduces to a sum rate maximization problem of downlink NOMA, whose decoding order is fixed as $1 \rightarrow K$ \cite{zhu:tvt:20}. Accordingly, our method to solve \eqref{eq : objective maximize align non coll} is also useful in downlink NOMA.
We explain this in detail later.
\end{remark}

\section{Precoding Optimization in the Non-Colluding Case}

We explain the ideas to solve the optimization problems \eqref{eq : objective maximize align non coll} in the non-colluding case.
First, to convert the problem into a tractable form, we approximate the minimum and maximum functions involved in \eqref{eq : objective maximize align non coll} using the LogSumExp technique. Subsequently, we represent the approximated objective function as a form of Rayleigh quotients by rewriting the precoding vector onto a higher dimensional vector. With this form, we derive the first-order optimality condition and put forth an efficient algorithm to find the principal eigenvector exploiting our novel interpretation of the derived optimality condition through a lens of a functional eigenvalue problem. 



\subsection{Reformulation to a Tractable Form}

First, we approximate the non-smooth minimum function by using the LogSumExp technique. With the technique, the minimum and maximum functions are approximated as \cite{shen:tpami:10,nielsen:en:16}
\begin{equation}\label{eq : min logsum}
    \min_{i=1,\cdots,N}\{x_i\} \approx -\frac{1}{\alpha}\log\left(\sum_{i=1}^N \exp(-x_i \alpha) \right),
\end{equation}
\begin{equation}\label{eq : max logsum}
    \max_{i=1,\cdots,N}\{x_i\} \approx \frac{1}{\alpha}\log\left(\sum_{i=1}^N \exp(x_i \alpha) \right),
\end{equation}
where the approximation becomes tight as $\alpha \rightarrow \infty$. Leveraging \eqref{eq : min logsum} and \eqref{eq : max logsum}, we obtain the following approximations 
\begin{equation}\label{eq : min rate}
    \min_{i\in \CMcal{L}_\ell, \ell\geq k}R_{k,i,\ell} \approx -\frac{1}{\alpha}\log\left(\sum_{\ell = k}^K\sum_{i\in \CMcal{L}_{\ell}}e^{-\alpha R_{k,i,\ell}} \right),
\end{equation}
\begin{equation}\label{eq : max rate}
    \max_{i'\in \CMcal{L}_{\ell'}, \ell' < k}R_{k,i',\ell'} \approx \frac{1}{\alpha}\log\left(\sum_{\ell'=1}^{k-1}\sum_{i'\in\CMcal{L}_{\ell'}}e^{\alpha R_{k,i',\ell'}} \right).
\end{equation}

Next, we define a higher dimensional precoding vector $\bar{\bF{f}}$ by stacking the original precoding vectors $\bF{f}_1,\cdots,\bF{f}_K$. Accordingly, $\bar{\bF{f}}$ is given by 
\begin{equation}\label{eq : new f}
    \bar{\bF{f}} = \left[\bF{f}_1^{\sf T},\cdots,\bF{f}_K^{\sf T} \right]^{\sf T}\in \bb{C}^{NK\times 1}.
\end{equation}
With \eqref{eq : new f}, we rewrite the \eqref{eq : achievable rate for message} as 
\begin{align}
    R_{k,i,\ell} &=\log_2\left(1+\frac{\left| \bF{h}_{i,\ell}^{\sf H}\bF{f}_k\right|^2}{\sum_{j = k+1}^K\left|\bF{h}_{i,\ell}^{\sf H}\bF{f}_j\right|^2 + \frac{\sigma^2}{P}}\right)= \log_2\left(\frac{\sum_{j = k}^K
    \bF{f}_j^{\sf H}\bF{h}_{i,\ell}\bF{h}_{i,\ell}^{\sf H}\bF{f}_j+\frac{\sigma^2}{P}}{\sum_{j = k+1}^K\bF{f}_j^{\sf H}\bF{h}_{i,\ell}\bF{h}_{i,\ell}^{\sf H}\bF{f}_j + \frac{\sigma^2}{P}}\right)\nonumber\\
    &=\log_2\left(\frac{\bar{\bF{f}}^{\sf H}\bF{A}_{k,i,\ell}\bar{\bF{f}}}{\bar{\bF{f}}^{\sf H}\bF{B}_{k,i,\ell}\bar{\bF{f}}}\right),\label{eq : Rayleigh quotient achievable rate}
\end{align}
where 
\begin{align}
    & \bF{A}_{k,i,\ell} = \text{blkdiag}\left(\bF{0}, \cdots,\bF{0},\underbrace{\bF{h}_{i,\ell}\bF{h}_{i,\ell}^{\sf H}}_{k\text{th block}},\cdots, \bF{h}_{i,\ell}\bF{h}_{i,\ell}^{\sf H}\right)
    +\frac{\sigma^2}{P}\bF{I}_{NK} \in \bb{C}^{NK\times NK}, \label{eq : effective channel matrix A} \\
    & \bF{B}_{k,i,\ell} = \bF{A}_{k,i,\ell} - \text{blkdiag}\left(\bF{0},\cdots,\bF{0},\underbrace{\bF{h}_{i,\ell}\bF{h}_{i,\ell}^{\sf H}}_{k\text{th block}},\bF{0},\cdots,\bF{0} \right)  \in\bb{C}^{NK\times NK}. \label{eq : effective channel matrix B}
\end{align}
Subsequently, leveraging \eqref{eq : Rayleigh quotient achievable rate}, \eqref{eq : min rate} is represented as
\begin{align} \label{eq : min rate Rayleigh}
    \min_{i\in \CMcal{L}_\ell, \ell\geq k}R_{k,i,\ell} &\approx -\frac{1}{\alpha}\log\left(\sum_{\ell = k}^K\sum_{i\in\CMcal{L}_{\ell}}e^{-\alpha R_{k,i,\ell}} \right) = -\frac{1}{\alpha}\log\left(\sum_{\ell = k}^K\sum_{i\in\CMcal{L}_{\ell}}e^{-\alpha \log_2\left(\frac{\bar{\bF{f}}^{\sf H}\bF{A}_{k,i,\ell}\bar{\bF{f}}}{\bar{\bF{f}}^{\sf H}\bF{B}_{k,i,\ell}\bar{\bF{f}}}\right)} \right)\nonumber\\
    &= -\frac{1}{\alpha}\log\left(\sum_{\ell = k}^K\sum_{i\in\CMcal{L}_{\ell}}\left(\frac{\bar{\bF{f}}^{\sf H}\bF{A}_{k,i,\ell}\bar{\bF{f}}}{\bar{\bF{f}}^{\sf H}\bF{B}_{k,i,\ell}\bar{\bF{f}}} \right)^{-\beta}\right),
\end{align}
where $\beta = \alpha\frac{1}{\log2}$.
Similar to this, \eqref{eq : max rate} is obtained as 
\begin{align} \label{eq : max rate Rayleigh}
        \max_{i'\in \CMcal{L}_\ell', \ell'< k}R_{k,i',\ell'} &\approx\frac{1}{\alpha}\log\left(\sum_{\ell' = 1}^{k-1}\sum_{i'\in\CMcal{L}_{\ell'}}e^{\alpha R_{k,i',\ell'}} \right) = \frac{1}{\alpha}\log\left(\sum_{\ell' = 1}^{k-1}\sum_{i'\in\CMcal{L}_{\ell'}}e^{\alpha \log_2\left(\frac{\bar{\bF{f}}^{\sf H}\bF{A}_{k,i',\ell'}\bar{\bF{f}}}{\bar{\bF{f}}^{\sf H}\bF{B}_{k,i',\ell'}\bar{\bF{f}}}\right)} \right)\nonumber\\
    &= \frac{1}{\alpha}\log\left(\sum_{\ell' = 1}^{k-1}\sum_{i'\in\CMcal{L}_{\ell'}}\left(\frac{\bar{\bF{f}}^{\sf H}\bF{A}_{k,i',\ell'}\bar{\bF{f}}}{\bar{\bF{f}}^{\sf H}\bF{B}_{k,i',\ell'}\bar{\bF{f}}} \right)^{\beta}\right).
\end{align}
Combining \eqref{eq : min rate Rayleigh} and \eqref{eq : max rate Rayleigh}, the secrecy rate of the message $s_k$ is approximately 
\begin{equation}\label{eq : information rate rayleigh quotient without ec}
    \bar{R}_k^{\sf nc} \approx -\frac{1}{\alpha}\log\left(\sum_{\ell = k}^K\sum_{i\in\CMcal{L}_{\ell}}\left(\frac{\bar{\bF{f}}^{\sf H}\bF{A}_{k,i,\ell}\bar{\bF{f}}}{\bar{\bF{f}}^{\sf H}\bF{B}_{k,i,\ell}\bar{\bF{f}}}\right)^{-\beta} \right)
    -\frac{1}{\alpha}\log\left(\sum_{\ell'=1}^{k-1}\sum_{i'\in\CMcal{L}_{\ell'}}\left(\frac{\bar{\bF{f}}^{\sf H}\bF{A}_{k,i',\ell'}\bar{\bF{f}}}{\bar{\bF{f}}^{\sf H}\bF{B}_{k,i',\ell'}\bar{\bF{f}}}\right)^{\beta}\right).
\end{equation}
Finally, the problem \eqref{eq : objective maximize align non coll} is reformulated as
\begin{align} 
    \max_{\bar{\bF{f}}}&\sum_{k=1}^K\left(-\frac{1}{\alpha}\log\left(\sum_{\ell = k}^K\sum_{i\in \CMcal{L}_\ell}\left(\frac{\bar{\bF{f}}^{\sf H}\bF{A}_{k,i,\ell}\bar{\bF{f}}}{\bar{\bF{f}}^{\sf H}\bF{B}_{k,i,\ell}\bar{\bF{f}}}\right)^{-\beta} \right)-\frac{1}{\alpha}\log\left(\sum_{\ell'=1}^{k-1}\sum_{i'\in \CMcal{L}_{\ell'}}\left(\frac{\bar{\bF{f}}^{\sf H}\bF{A}_{k,i',\ell'}\bar{\bF{f}}}{\bar{\bF{f}}^{\sf H}\bF{B}_{k,i',\ell'}\bar{\bF{f}}}\right)^{\beta}\right)\right).\label{eq : objective rayleigh quotient non coll}
\end{align}
We note that the reformulated problem \eqref{eq : objective rayleigh quotient non coll} does not contain the transmit power constraint. This is because the objective function in \eqref{eq : objective rayleigh quotient non coll} is presented as a form of Rayleigh quotients.
Therefore 
the power constraint 
$\| \bar {\bf{f}} \|$ can be normalized in both of nominator and the denominator without affecting the approximated objective function. 
Now we are ready to tackle the problem in \eqref{eq : objective rayleigh quotient non coll}.
    
\subsection{First-Order KKT Condition}
In order to get an insight on the solution of the approximate problem \eqref{eq : objective rayleigh quotient non coll}, we drive a first-order KKT condition. The following lemma shows the main result in this subsection.
\begin{lemma} \label{lem:noncoll}
    In the non-colluding case, the first-order KKT condition of the  optimization problem \eqref{eq : objective rayleigh quotient non coll} is satisfied if the following holds.
    \begin{equation}\label{eq : eigenvalue problem non coll}
        {\bF{B}}_{\sf nc}^{-1}(\bar{\bF{f}}){\bF{A}}_{\sf nc}(\bar{\bF{f}})\bar{\bF{f}} = \lambda_{\sf nc}(\bar{\bF{f}})\bar{\bF{f}},
    \end{equation}
    where
    \begin{multline}\label{eq : A non coll}
        {\bF{A}}_{\sf nc}(\bar{\bF{f}}) \!=\!\lambda_{\sf nc}^{(\sf{num})}(\bar{\bF{f}})\!\times\!\sum_{k=1}^K\left(\frac{\sum\limits_{\ell=k}^{K}\sum\limits_{i \in \CMcal{L}_{\ell}}\left(\beta\left(\frac{\bar{\bF{f}}^{\sf H}\bF{A}_{k,i,\ell}\bar{\bF{f}}}{\bar{\bF{f}}^{\sf H}\bF{B}_{k,i,\ell}\bar{\bF{f}}}\right)^{-\beta}\frac{\bF{A}_{k,i,\ell}}{\bar{\bF{f}}^{\sf H}\bF{A}_{k,i,\ell}\bar{\bF{f}}} \right)}{\sum\limits_{\ell=k}^{K}\sum\limits_{i \in \CMcal{L}_{\ell}}\left(\frac{\bar{\bF{f}}^{\sf H}\bF{A}_{k,i,\ell}\bar{\bF{f}}}{\bar{\bF{f}}^{\sf H}\bF{B}_{k,i,\ell}\bar{\bF{f}}} \right)^{-\beta}} +\frac{\sum\limits_{\ell=1}^{k-1}\!\sum\limits_{i \in \CMcal{L}_{\ell}}\!\left(\beta\left(\frac{\bar{\bF{f}}^{\sf H}\!\bF{A}_{k,i,\ell}\bar{\bF{f}}}{\bar{\bF{f}}^{\sf H}\!\bF{B}_{k,i,\ell}\bar{\bF{f}}}\right)^{\beta}\frac{\bF{B}_{k,i,\ell}}{\bar{\bF{f}}^{\sf H}\!\bF{B}_{k,i,\ell}\bar{\bF{f}}} \right)}{\sum\limits_{\ell=1}^{k-1}\sum\limits_{i \in \CMcal{L}_{\ell}}\left(\frac{\bar{\bF{f}}^{\sf H}\bF{A}_{k,i,\ell}\bar{\bF{f}}}{\bar{\bF{f}}^{\sf H}\bF{B}_{k,i,\ell}\bar{\bF{f}}} \right)^\beta} \right),
    \end{multline}
    \begin{multline}\label{eq : B non coll}
        {\bF{B}}_{\sf nc}(\bar{\bF{f}}) \!=\!\lambda_{\sf nc}^{(\sf{den})}(\bar{\bF{f}})\times\sum_{k=1}^K\left(\frac{\sum\limits_{\ell=k}^{K}\sum\limits_{i \in \CMcal{L}_{\ell}}\left(\beta\left(\frac{\bar{\bF{f}}^{\sf H}\bF{A}_{k,i,\ell}\bar{\bF{f}}}{\bar{\bF{f}}^{\sf H}\bF{B}_{k,i,\ell}\bar{\bF{f}}}\right)^{-\beta}\frac{\bF{B}_{k,i,\ell}}{\bar{\bF{f}}^{\sf H}\bF{B}_{k,i,\ell}\bar{\bF{f}}} \right)}{\sum\limits_{\ell=k}^{K}\sum\limits_{i \in \CMcal{L}_{\ell}}\left(\frac{\bar{\bF{f}}^{\sf H}\bF{A}_{k,i,\ell}\bar{\bF{f}}}{\bar{\bF{f}}^{\sf H}\bF{B}_{k,i,\ell}\bar{\bF{f}}} \right)^{-\beta}}+\frac{\sum\limits_{\ell=1}^{k-1}\!\sum\limits_{i \in \CMcal{L}_{\ell}}\!\left(\!\beta\!\left(\frac{\bar{\bF{f}}^{\sf H}\!\bF{A}_{k,i,\ell}\bar{\bF{f}}}{\bar{\bF{f}}^{\sf H}\!\bF{B}_{k,i,\ell}\bar{\bF{f}}}\right)^{\beta}\!\frac{\bF{A}_{k,i,\ell}}{\bar{\bF{f}}^{\sf H}\!\bF{A}_{k,i,\ell}\bar{\bF{f}}} \right)}{\sum\limits_{\ell=1}^{k-1}\sum\limits_{i \in \CMcal{L}_{\ell}}\left(\frac{\bar{\bF{f}}^{\sf H}\bF{A}_{k,i,\ell}\bar{\bF{f}}}{\bar{\bF{f}}^{\sf H}\bF{B}_{k,i,\ell}\bar{\bF{f}}} \right)^\beta} \right),
    \end{multline}
    with
    \begin{align}\label{eq : lambda non coll}
     \lambda_{\sf nc}(\bar{\bF{f}})& = \sum_{k=1}^K\left(-\frac{1}{\alpha}\log\left(\sum_{\ell = k}^K\sum_{i\in \CMcal{L}_{\ell}}\left(\frac{\bar{\bF{f}}^{\sf H}\bF{A}_{k,i,\ell}\bar{\bF{f}}}{\bar{\bF{f}}^{\sf H}\bF{B}_{k,i,\ell}\bar{\bF{f}}}\right)^{-\beta} \right)
     -\frac{1}{\alpha}\log\left(\sum_{\ell'=1}^{k-1}\sum_{i'\in\CMcal{L}_{\ell'}}\left(\frac{\bar{\bF{f}}^{\sf H}\bF{A}_{k,i',\ell'}\bar{\bF{f}}}{\bar{\bF{f}}^{\sf H}\bF{B}_{k,i',\ell'}\bar{\bF{f}}}\right)^{\beta}\right)\right)\\
     &= \frac{\lambda_{\sf nc}^{(\sf{num})}(\bar{\bF{f}})}{\lambda_{\sf nc}^{(\sf {den})}(\bar{\bF{f}})}.
    \end{align}
\end{lemma}
\begin{proof}
See Appendix A.1.
\end{proof}
Now we interpret the derived optimality condition \eqref{eq : eigenvalue problem non coll}. If the precoding vector $\bar {\bf{f}}$ satisfies \eqref{eq : eigenvalue problem non coll}, it satisfies the first-order optimality condition, implying that $\bar {\bf{f}}$ is located on a stationary point that has zero-gradient. This, however, does not guarantee that $\bar {\bf{f}}$ is a good solution. Among such points, we need to find a local optimal point that maximizes the objective function of \eqref{eq : objective rayleigh quotient non coll}. To this end, we first observe that \eqref{eq : eigenvalue problem non coll} can be cast as a generalized eigenvalue problem. Rigorously, \eqref{eq : eigenvalue problem non coll} is interpreted as a class of a eigenvector dependent nonlinear eigenvalue problem (NEPv) \cite{cai:sjmaa:18}. A distinguishable feature of NEPv compared to a typical eigenvalue problem is that a matrix is a function of eigenvectors. Based on this interpretation, ${\bF{B}}_{\sf nc}^{-1}(\bar{\bF{f}}){\bF{A}_{\sf nc}}(\bar{\bF{f}})$ behaves as the corresponding matrix, $\bar{\bF{f}}$ behaves as the eigenvector of the matrix ${\bF{B}}_{\sf nc}^{-1}(\bar{\bF{f}}){\bF{A}_{\sf nc}}(\bar{\bF{f}})$, and $\lambda_{\sf nc}(\bar{\bF{f}})$ behaves as the eigenvalue. Noticeably, the eigenvalue $\lambda_{\sf nc}(\bar {\bf{f}})$ is equivalent with the objective function that we want to maximize. For this reason, if we find the principal eigenvector of ${\bF{B}}_{\sf nc}^{-1}(\bar{\bF{f}}){\bF{A}_{\sf nc}}(\bar{\bF{f}})$, then it maximizes our objective function of \eqref{eq : objective rayleigh quotient non coll} while satisfying \eqref{eq : eigenvalue problem non coll}. Consequently, finding the principal eigenvector of ${\bF{B}}_{\sf nc}^{-1}(\bar{\bF{f}}){\bF{A}_{\sf nc}}(\bar{\bF{f}})$ is equivalent to finding the best local optimal point. 

Finding the principal eigenvector of ${\bF{B}}_{\sf nc}^{-1}(\bar{\bF{f}}){\bF{A}_{\sf nc}}(\bar{\bF{f}})$ is far from trivial. A main challenge comes from that the matrix ${\bF{B}}_{\sf nc}^{-1}(\bar{\bF{f}}){\bF{A}_{\sf nc}}(\bar{\bF{f}})$ changes depending on $\bar {\bf{f}}$. In the next subsection, we propose a novel method called GPI-HIA for the non-colluding case (GPI-HIA (Non-Coll)) so as to efficiently obtain the principal eigenvector. 

\subsection{GPI-HIA for the Non-Colluding Case}
We present the GPI-HIA (Non-Coll) method. Inspired from the typical power iteration method, the proposed method iteratively updates the precoding vector $\bar {\bf{f}}$ as 
\begin{equation}
\bar{\bF{f}}_{(t)} \leftarrow \frac{\bF{B}_{\sf nc}^{-1}(\bar{\bF{f}}_{(t-1)})\bF{A}_{\sf nc}(\bar{\bF{f}}_{(t-1)})\bar{\bF{f}}_{(t-1)}}{\left\|\bF{B}_{\sf nc}^{-1}(\bar{\bF{f}}_{(t-1)})\bF{A}_{\sf nc}(\bar{\bF{f}}_{(t-1)})\bar{\bF{f}}_{(t-1)} \right\|}.
\end{equation}
We repeat this process until the convergence criterion is met. The convergence condition is $\left\|\bar{\bF{f}}_{(t)}-\bar{\bF{f}}_{(t-1)} \right\|<\epsilon$, where $\epsilon$ is the tolerance level. Algorithm 1 summarizes the process.
\begin{algorithm} [t]
\caption{GPI-HIA (Non-Coll) } \label{alg:main nc} 
 {\bf{initialize}}: $\bar {\bf{f}}_{(0)}$ (MRT) \\
 Set the iteration count $t = 1$ \\
\While {$\left\|\bar {\bf{f}}_{(t)} - \bar {\bf{f}}_{(t-1)} \right\| > \epsilon$}{
 Construct the matrix ${\bF{A}}_{\sf nc}(\bar {\bF{f}}_{(t-1)})$ by using \eqref{eq : A non coll}. \\
 Construct the matrix ${\bf{B}}_{\sf nc} (\bar {\bF{f}}_{(t-1)})$ by using \eqref{eq : B non coll}. \\
 Compute $\bar {\bF{f}}_{(t)} = {\bF{B}}_{\sf nc}^{-1}(\bar{\bF{f}}_{(t-1)}){\bF{A}}_{\sf nc}(\bar{\bF{f}}_{(t-1)})\bar{\bF{f}}_{(t-1)}$. \\
 Normalize $\bar{\bf{f}}_{(t)} = \bar {\bf{f}}_{(t)}/\left\| \bar {\bf{f}}_{(t)}\right\|$.\\
 $t \leftarrow t+1$.}
 \Return{\ }{${\bf \bar f} = {\bf \bar f}_{(t)}$}.\\
\end{algorithm}

\begin{remark} \normalfont ($\alpha$ tuning)
The parameter $\alpha$ determines the accuracy of the LogSumExp approximation. As we use large $\alpha$, the approximation becomes tight; thereby using large $\alpha$ is desirable. 
Nevertheless, as presented in \cite{park:arxiv:21}, too large $\alpha$ may make the proposed GPI-HIA algorithm not converge. This is because, as $\alpha$ increases, the LogSumExp function becomes more similar to the minimum function and its shape turns to be non-smooth. In this case, no stationary point is characterized; thus, the algorithm cannot find the converging point. 

To properly tune $\alpha$, we use the following method. We start the GPI-HIA with large $\alpha$. If the iteration loop of the GPI-HIA does not converge within the predetermined number of iterations, we regard that the used $\alpha$ is too large to make the algorithm converge, so that we enforce to terminate the loop, decrease $\alpha$ with the predetermined amount, and newly start the algorithm again. We repeat this process until the algorithm converges before the predetermined number. As shown in the later section, the GPI-HIA algorithm with this $\alpha$ tuning performs very well. 


\end{remark}

\begin{remark}[Complexity and implementation] \normalfont
Similar to \cite{choi:arxiv:21, park:arxiv:21, choi:twc:20}, the computational complexity is dominated by calculating ${\bf{B}}_{\sf nc}^{-1}(\bar {\bf{f}})$. Since ${\bf{B}}_{\sf nc}(\bar {\bf{f}})$ is a block-diagonal matrix that consists of $K$ number of $N\times N$ sub-matrices, its inverse can be obtained by acquiring the inversion of each sub-matrix. Accordingly, the computational complexity per iteration is analyzed as $\CMcal{O}(K N^{3})$. 

We emphasize on that the proposed method has a benefit in implementation. For example, in \cite{zhang:tvt:19} where a transmit power minimization for the HIA model was tackled, a non-convex original problem was relaxed into a convex form and CVX was exploited to get a solution. Nevertheless, since CVX is not designed to run in real-time FPGA hardware \cite{jevganij:access:19}, it is infeasible to use this convex relaxation-based method in practice. 
On the contrary, our method does not need to use any off-the-shelf solver including CVX. The only required computational load is the matrix inversion, which is also used in a very simple precoding strategy such as ZF. For this reason, our method is preferable in terms of implementation compared to conventional convex relaxation based approaches. 
\end{remark}


\subsection{Fairness}
One possible issue in the HIA model is fairness. Since we maximize the sum secrecy rate, users in a certain layer may suffer from a very low rate. For example, increasing the power of the high priority messages incurs critical interference to the lower priority users, while increasing the power of the low priority messages is not harmful because the high priority users can eliminate the low priority messages via SIC. For this reason, in a perspective of the sum secrecy rate maximization, it is beneficial to increase the power of the low priority messages and decrease the secrecy rate of the high priority users, which can lead to an undesirable rate distribution in terms of fairness.

To resolve this fairness issue, we modify the proposed GPI-HIA (Non-Coll) by adopting the proportional fairness (PF) policy \cite{viswanath:tit:02}. In the PF algorithm, the BS traces the previously served rate on average and reflects this value into the optimization problem as inverse weights. 
Specifically, denoting that $\bar R_k^{\sf nc}(t)$ as the achieved secrecy rate for the message $s_k$ in the transmission block $t$, the average secrecy rate $\mu_{k}^{\sf nc}(t+1)$ is updated by a simple first-order autoregressive filter as
\begin{equation}\label{eq : mu non coll}
    \mu_{k}^{\sf nc}(t+1) = (1-\delta)\mu_{k}^{\sf nc}(t)+\delta\bar{R}_k^{\sf nc}(t),
\end{equation}
where $\delta \in (0,1)$ is a given parameter. With $\mu_{k}^{\sf nc}(t)$, we modify the problem \eqref{eq : objective maximize align non coll} into the weighted sum secrecy rate maximization problem as 
\begin{align}
    \max_{\bF{f}_1,\cdots,\bF{f}_K}&\sum_{k=1}^K\frac{1}{\mu_k^{\sf nc}(t)}\bar{R}_k^{\sf nc}\label{eq : optimization problem non-coll with pf}\\
    \text{subject to}&\sum_{k=1}^{K}\|\bF{f}_k \|^2 = 1.
\end{align}
By doing this, if the message $s_k$ has obtained very small secrecy rate during the previous transmission periods, $\mu_k^{\sf nc}(t)$ decreases. This leads to  the increase in the weight of $\bar R_k^{\sf nc}$. Then the BS tends to increase $\bar R_k^{\sf nc}$ considerably in the next transmission period. On the contrary to that, if the message $s_k$ has obtained large secrecy rate during the previous transmission periods, the PF algorithm puts a little efforts into increasing $\bar R_k^{\sf nc}$, so as to provide the secrecy rate fairly. 

To obtain a solution of the modified weighted sum secrecy rate maximization problem \eqref{eq : optimization problem non-coll with pf}, the corresponding first-order optimality condition is derived as follows. 
\begin{corollary} \label{coro : non-coll with pf}
    In the non-colluding case, the first-order KKT condition of the modified optimization problem \eqref{eq : optimization problem non-coll with pf} is satisfied if the following holds.
    \begin{equation}\label{eq : eigenvalue problem non coll with pf}
        \bar{\bF{B}}_{\sf nc}^{-1}(\bar{\bF{f}})\bar{\bF{A}}_{\sf nc}(\bar{\bF{f}})\bar{\bF{f}} = \bar{\lambda}_{\sf nc}(\bar{\bF{f}})\bar{\bF{f}},
    \end{equation}
    where
\begin{equation}\label{eq : A with pf non coll}
    \bar{\bF{A}}_{\sf nc}(\bar{\bF{f}}) \!=\! \bar{\lambda}^{\sf (num)}_{\sf nc}(\bar{\bF{f}})\!\times\!\sum_{k=1}^K\frac{1}{\mu_{k}^{\sf nc}(t)}\left(\frac{\sum\limits_{\ell=k}^{K}\sum\limits_{i \in \CMcal{L}_{\ell}}\left(\beta\left(\frac{\bar{\bF{f}}^{\sf H}\bF{A}_{k,i,\ell}\bar{\bF{f}}}{\bar{\bF{f}}^{\sf H}\bF{B}_{k,i,\ell}\bar{\bF{f}}}\right)^{-\beta}\frac{\bF{A}_{k,i,\ell}}{\bar{\bF{f}}^{\sf H}\bF{A}_{k,i,\ell}\bar{\bF{f}}} \right)}{\sum\limits_{\ell=k}^{K}\sum\limits_{i \in \CMcal{L}_{\ell}}\left(\frac{\bar{\bF{f}}^{\sf H}\bF{A}_{k,i,\ell}\bar{\bF{f}}}{\bar{\bF{f}}^{\sf H}\bF{B}_{k,i,\ell}\bar{\bF{f}}} \right)^{-\beta}} +\frac{\sum\limits_{\ell=1}^{k-1}\!\sum\limits_{i \in \CMcal{L}_{\ell}}\!\left(\!\beta\!\left(\frac{\bar{\bF{f}}^{\sf H}\!\bF{A}_{k,i,\ell}\bar{\bF{f}}}{\bar{\bF{f}}^{\sf H}\!\bF{B}_{k,i,\ell}\bar{\bF{f}}}\right)^{\beta}\frac{\bF{B}_{k,i,\ell}}{\bar{\bF{f}}^{\sf H}\!\bF{B}_{k,i,\ell}\bar{\bF{f}}} \right)}{\sum\limits_{\ell=1}^{k-1}\sum\limits_{i \in \CMcal{L}_{\ell}}\left(\frac{\bar{\bF{f}}^{\sf H}\bF{A}_{k,i,\ell}\bar{\bF{f}}}{\bar{\bF{f}}^{\sf H}\bF{B}_{k,i,\ell}\bar{\bF{f}}} \right)^\beta} \right),
\end{equation}
\begin{equation}\label{eq : B with pf non coll}
    \bar{{\bF{B}}}_{\sf nc}(\bar{\bF{f}}) \!=\!\bar{\lambda}_{\sf nc}^{\sf{(den)}}(\bar{\bF{f}})\!\times\!\sum_{k=1}^K\frac{1}{\mu_{k}^{\sf nc}(t)}\left(\frac{\sum\limits_{\ell=k}^{K}\sum\limits_{i \in \CMcal{L}_{\ell}}\left(\beta\left(\frac{\bar{\bF{f}}^{\sf H}\bF{A}_{k,i,\ell}\bar{\bF{f}}}{\bar{\bF{f}}^{\sf H}\bF{B}_{k,i,\ell}\bar{\bF{f}}}\right)^{-\beta}\frac{\bF{B}_{k,i,\ell}}{\bar{\bF{f}}^{\sf H}\bF{B}_{k,i,\ell}\bar{\bF{f}}} \right)}{\sum\limits_{\ell=k}^{K}\sum\limits_{i \in \CMcal{L}_{\ell}}\left(\frac{\bar{\bF{f}}^{\sf H}\bF{A}_{k,i,\ell}\bar{\bF{f}}}{\bar{\bF{f}}^{\sf H}\bF{B}_{k,i,\ell}\bar{\bF{f}}} \right)^{-\beta}}+\frac{\sum\limits_{\ell=1}^{k-1}\!\sum\limits_{i \in \CMcal{L}_{\ell}}\!\left(\!\beta\!\left(\frac{\bar{\bF{f}}^{\sf H}\!\bF{A}_{k,i,\ell}\bar{\bF{f}}}{\bar{\bF{f}}^{\sf H}\!\bF{B}_{k,i,\ell}\bar{\bF{f}}}\right)^{\beta}\!\frac{\bF{A}_{k,i,\ell}}{\bar{\bF{f}}^{\sf H}\!\bF{A}_{k,i,\ell}\bar{\bF{f}}} \right)}{\sum\limits_{\ell=1}^{k-1}\sum\limits_{i \in \CMcal{L}_{\ell}}\left(\frac{\bar{\bF{f}}^{\sf H}\bF{A}_{k,i,\ell}\bar{\bF{f}}}{\bar{\bF{f}}^{\sf H}\bF{B}_{k,i,\ell}\bar{\bF{f}}} \right)^\beta}\right),
\end{equation}
and
\begin{align}\label{eq : lambda with pf non coll}
     \bar{\lambda}_{\sf nc}(\bar{\bF{f}})& = \sum_{k=1}^K\frac{1}{\mu_{k}^{\sf nc}(t)}\left(-\frac{1}{\alpha}\log\left(\sum_{\ell = k}^K\sum_{i\in \CMcal{L}_{\ell}}\left(\frac{\bar{\bF{f}}^{\sf H}\bF{A}_{k,i,\ell}\bar{\bF{f}}}{\bar{\bF{f}}^{\sf H}\bF{B}_{k,i,\ell}\bar{\bF{f}}}\right)^{-\beta} \right)
     -\frac{1}{\alpha}\log\left(\sum_{\ell'=1}^{k-1}\sum_{i'\in\CMcal{L}_{\ell'}}\left(\frac{\bar{\bF{f}}^{\sf H}\bF{A}_{k,i',\ell'}\bar{\bF{f}}}{\bar{\bF{f}}^{\sf H}\bF{B}_{k,i',\ell'}\bar{\bF{f}}}\right)^{\beta}\right)\right)\\
     &= \frac{\bar{\lambda}_{\sf nc}^{\sf{(num)}}(\bar{\bF{f}})}{\bar{\lambda}_{\sf nc}^{\sf {(den)}}(\bar{\bF{f}})}.
\end{align}

\end{corollary}
\begin{proof}
 The proof is straightforward by adding the weight $1/\mu_k^{\sf nc}(t)$ to the partial derivative of $\lambda_{\sf nc}(\bar{\bF{f}})$ in Appendix A.1.
\end{proof}
Using Corollary \ref{coro : non-coll with pf}, the proposed GPI-HIA (Non-Coll) with the PF policy follows similar steps to Algorithm \ref{alg:main nc}: 
Step 1. we compute the precoding vector using GPI-HIA, i.e., the pincipal eigenvector of $\bar{\bF{B}}_{\sf nc}^{-1}(\bar{\bF{f}})\bar{\bF{A}}_{\sf nc}(\bar{\bF{f}})$, 
Step 2. update $\mu_{k}^{\sf nc}(t)$ according to \eqref{eq : mu non coll}, and Step 3. repeat the steps 1 to 2. 
In the later section, we show that the proposed GPI-HIA adopting the PF policy achieves much improved fairness. 



\section{Precoding Optimization in the Colluding Case}
In this section, we consider the colluding case assuming that the lower priority users cooperate to decode the high priority message. To solve the optimization problem in \eqref{eq : objective maximize align coll} which is formulated for the colluding case, similar to the non-colluding case, we first convert the minimum function involved in \eqref{eq : objective maximize align coll} using the LogSumExp technique. Subsequently, we reformulate the approximated problem as a form of Rayleigh quotients with regard to the high dimensional vector \eqref{eq : new f}. With this form, we drive the first-order optimality condition and propose a novel algorithm to obtain the principal eigenvector.

\subsection{Reformulation to a Tractable Form}

A difference of the colluding case compared to the non-colluding case is on the achievable rate of the wiretapping channel. The achievable rate of the colluding wiretapping channel is 
\begin{align}
    R_{k,i',\ell'}^{\sf e} &= \log_2\left(1+\sum_{\ell' = 1}^{k-1}\sum_{i'\in \Cal{L}_{\ell'}}\frac{\left| \bF{h}_{i',\ell'}^{\sf H}\bF{f}_k\right|^2}{\sum_{j = k+1}^{K}\left|\bF{h}_{i',\ell'}^{\sf H}\bF{f}_j\right|^2 + \frac{\sigma^2}{P}} \right)\nonumber\\
    &= \log_2\left(\sum_{\ell' = 1}^{k-1}\sum_{i'\in\CMcal{L}_\ell'}\left(\frac{ \sum_{j=k}^K\bF{f}_j^{\sf H}\bF{h}_{i',\ell'}\bF{h}_{i',\ell'}^{\sf H}\bF{f}_j+\frac{\sigma^2}{P}}{\sum_{j = k+1}^{K}\bF{f}_j^{\sf H}\bF{h}_{i',\ell'}\bF{h}_{i',\ell'}^{\sf H}\bF{f}_j + \frac{\sigma^2}{P}} + \frac{1}{\gamma_k}\right) \right) 
    =\frac{1}{\log2}\log\left(\sum_{\ell' = 1}^{k-1}\sum_{i'\in\CMcal{L}_{\ell'}}\left(\frac{\bar{\bF{f}}^{\sf H}\bF{C}_{k,i',\ell'}\bar{\bF{f}}}{\bar{\bF{f}}^{\sf H}\bF{D}_{k,i',\ell'}\bar{\bF{f}}} \right) \right),\label{eq : rayleigh quotient wiretapping channel}
\end{align}
where 
\begin{equation}\label{eq : effective channel matrix C}
    \bF{C}_{k,i',\ell'} = \text{blkdiag}\left(\bF{0},\cdots,\bF{0},\underbrace{\gamma_k\bF{h}_{i',\ell'}\bF{h}_{i',\ell'}^{\sf H}}_{k\text{th block}},\bF{h}_{i',\ell'}\bF{h}_{i',\ell'}^{\sf H},\cdots,\bF{h}_{i',\ell'}\bF{h}_{i',\ell'}^{\sf H} \right)+\frac{\sigma^2}{P}\bF{I}_{NK},
\end{equation}
\begin{equation}\label{eq : effective channel matrix D}
    \bF{D}_{k,i',\ell'} = \gamma_k\times\text{blkdiag}\left(\bF{0},\cdots,\bF{0},\underbrace{\bF{h}_{i',\ell'}\bF{h}_{i',\ell'}^{\sf H}}_{(k+1)\text{th block}},\cdots,\bF{h}_{i',\ell'}\bF{h}_{i',\ell'}^{\sf H} \right)+\frac{\gamma_k\sigma^2}{P}\bF{I}_{NK}.
\end{equation}
and $\gamma_k = \sum_{\ell'=1}^{k-1}|\CMcal{L}_{\ell'}|$. Combining \eqref{eq : min rate Rayleigh} and \eqref{eq : rayleigh quotient wiretapping channel}, the secrecy rate of the message $s_k$ is approximately 
\begin{equation}\label{eq : secrecy rate rayleigh quotient with ec}
    \bar{R}_k^{\sf c} \approx -\frac{1}{\alpha}\log\left(\sum_{\ell=k}^K\sum_{i\in\CMcal{L}_\ell}\left(\frac{\bar{\bF{f}}^{\sf H}\bF{A}_{k,i,\ell}\bar{\bF{f}}}{\bar{\bF{f}}^{\sf H}\bF{B}_{k,i,\ell}\bar{\bF{f}}} \right)^{-\beta}\right)-\frac{1}{\log 2}\log\left(\sum_{\ell' = 1}^{k-1}\sum_{i'\in\CMcal{L}_{\ell'}}\left(\frac{\bar{\bF{f}}^{\sf H}\bF{C}_{k,i',\ell'}\bar{\bF{f}}}{\bar{\bF{f}}^{\sf H}\bF{D}_{k,i',\ell'}\bar{\bF{f}}} \right) \right).
\end{equation}
Finally, with these high dimensional representations, the problem \eqref{eq : objective maximize align coll} is reformulated as
\begin{align}
    \max_{\bar{\bF{f}}}&\sum_{k=1}^K\left( -\frac{1}{\alpha}\log\left(\sum_{\ell=k}^K\sum_{i\in\CMcal{L}_\ell}\left(\frac{\bar{\bF{f}}^{\sf H}\bF{A}_{k,i,\ell}\bar{\bF{f}}}{\bar{\bF{f}}^{\sf H}\bF{B}_{k,i,\ell}\bar{\bF{f}}} \right)^{-\beta}\right)-\frac{1}{\log 2}\log\left(\sum_{\ell' = 1}^{k-1}\sum_{i'\in\CMcal{L}_{\ell'}}\left(\frac{\bar{\bF{f}}^{\sf H}\bF{C}_{k,i',\ell'}\bar{\bF{f}}}{\bar{\bF{f}}^{\sf H}\bF{D}_{k,i',\ell'}\bar{\bF{f}}} \right) \right)\right).\label{eq : objective rayleigh quotient coll}
\end{align}
The transmit power constraint vanishes in \eqref{eq : objective rayleigh quotient coll} because when reformulating the objective equation into form of Rayleigh quotients as \eqref{eq : objective rayleigh quotient coll}, numerator and denominator are normalized by $\left\|\bar{\bF{f}} \right\|$ without affecting the objective equation. Now we are ready to tackle the problem \eqref{eq : objective rayleigh quotient coll}.

\subsection{First-Order KKT Condition}
In order to obtain a solution of \eqref{eq : objective rayleigh quotient coll}, we drive a first-order KKT condition. The following lemma shows the main result in this subsection.
\begin{lemma}
In the colluding case, the first-order KKT condition of the optimization problem \eqref{eq : objective rayleigh quotient coll} is satisfied if the following holds.
\begin{equation}\label{eq : eigenvalue problem coll}
    \bF{B}_{\sf c}^{-1}(\bar{\bF{f}})\bF{A}_{\sf c}(\bar{\bF{f}}) = \lambda_{\sf c}(\bar{\bF{f}})\bar{\bF{f}}
\end{equation}
where
\begin{equation}\label{eq : A coll}
    \bF{A}_{\sf c}(\bar{\bF{f}})\! =\! \lambda_{\sf c}^{\sf (num)}\times\!\sum_{k = 1}^K\!\left(\frac{1}{\alpha}\frac{\sum\limits_{\ell = k}^K\sum\limits_{i\in\CMcal{L}_\ell}\!\beta\!\left(\frac{\bar{\bF{f}}^{\sf H}\bF{A}_{k,i,\ell}\bar{\bF{f}}}{\bar{\bF{f}}^{\sf H}\bF{B}_{k,i,\ell}\bar{\bF{f}}} \right)^{-\beta}\!\frac{\bF{A}_{k,i,\ell}}{\bar{\bF{f}}^{\sf H}\bF{A}_{k,i,\ell}\bar{\bF{f}}}}{\sum\limits_{\ell = k}^K\sum\limits_{i\in\Cal{L}_\ell}\left(\frac{\bar{\bF{f}}^{\sf H}\bF{A}_{k,i,\ell}\bar{\bF{f}}}{\bar{\bF{f}}^{\sf H}\bF{B}_{k,i,\ell}\bar{\bF{f}}}\right)^{-\beta}} + \frac{1}{\log2}\frac{\sum\limits_{\ell' = 1}^{k-1}\sum\limits_{i'\in\Cal{L}_{\ell'}}\frac{\bar{\bF{f}}^{\sf H}\bF{C}_{k,i',\ell'}\bar{\bF{f}}}{\bar{\bF{f}}^{\sf H}\bF{D}_{k,i',\ell'}\bar{\bF{f}}}\left(\frac{\bF{D}_{k,i',\ell'}}{\bar{\bF{f}}^{\sf H}\bF{D}_{k,i',\ell'}\bar{\bF{f}}} \right)}{\sum\limits_{\ell' = 1}^{k-1}\sum\limits_{i'\in\Cal{L}_{\ell'}}\frac{\bar{\bF{f}}^{\sf H}\bF{C}_{k,i',\ell'}\bar{\bF{f}}}{\bar{\bF{f}}^{\sf H}\bF{D}_{k,i',\ell'}\bar{\bF{f}}}}\right),
\end{equation}
\begin{equation}\label{eq : B coll}
    \bF{B}_{\sf c}(\bar{\bF{f}})\! =\! \lambda_{\sf c}^{\sf (den)}\times\!\sum_{k = 1}^K\!\left(\frac{1}{\alpha}\frac{\sum\limits_{\ell = k}^K\sum\limits_{i\in\Cal{L}_\ell}\!\beta\!\left(\frac{\bar{\bF{f}}^{\sf H}\bF{A}_{k,i,\ell}\bar{\bF{f}}}{\bar{\bF{f}}^{\sf H}\bF{B}_{k,i,\ell}\bar{\bF{f}}} \right)^{-\beta}\!\frac{\bF{B}_{k,i,\ell}}{\bar{\bF{f}}^{\sf H}\bF{B}_{k,i,\ell}\bar{\bF{f}}}}{\sum\limits_{\ell = k}^K\sum\limits_{i\in\Cal{L}_\ell}\left(\frac{\bar{\bF{f}}^{\sf H}\bF{A}_{k,i,\ell}\bar{\bF{f}}}{\bar{\bF{f}}^{\sf H}\bF{B}_{k,i,\ell}\bar{\bF{f}}}\right)^{-\beta}} + \frac{1}{\log2}\frac{\sum\limits_{\ell' = 1}^{k-1}\sum\limits_{i'\in\Cal{L}_{\ell'}}\frac{\bar{\bF{f}}^{\sf H}\bF{C}_{k,i',\ell'}\bar{\bF{f}}}{\bar{\bF{f}}^{\sf H}\bF{D}_{k,i',\ell'}\bar{\bF{f}}}\left(\frac{\bF{C}_{k,i',\ell'}}{\bar{\bF{f}}^{\sf H}\bF{C}_{k,i',\ell'}\bar{\bF{f}}} \right)}{\sum\limits_{\ell' = 1}^{k-1}\sum\limits_{i'\in\Cal{L}_{\ell'}}\frac{\bar{\bF{f}}^{\sf H}\bF{C}_{k,i',\ell'}\bar{\bF{f}}}{\bar{\bF{f}}^{\sf H}\bF{D}_{k,i',\ell'}\bar{\bF{f}}}}\right),    
\end{equation}
with 
\begin{align}\label{eq : lambda coll}
    \lambda_{\sf c}(\bar{\bF{f}})& = \sum_{k=1}^K\left( -\frac{1}{\alpha}\log\left(\sum_{\ell=k}^K\sum_{i\in\CMcal{L}_\ell}\left(\frac{\bar{\bF{f}}^{\sf H}\bF{A}_{k,i,\ell}\bar{\bF{f}}}{\bar{\bF{f}}^{\sf H}\bF{B}_{k,i,\ell}\bar{\bF{f}}} \right)^{-\beta}\right)-\frac{1}{\log 2}\log\left(\sum_{\ell' = 1}^{k-1}\sum_{i'\in\CMcal{L}_{\ell'}}\left(\frac{\bar{\bF{f}}^{\sf H}\bF{C}_{k,i',\ell'}\bar{\bF{f}}}{\bar{\bF{f}}^{\sf H}\bF{D}_{k,i',\ell'}\bar{\bF{f}}} \right) \right) \right)\\
    &= \frac{\lambda_{\sf c}^{\sf (num)}(\bar{\bF{f}})}{\lambda^{\sf (den)}_{\sf c}(\bar{\bF{f}})}.
\end{align}
\end{lemma}
\begin{proof}
See Appendix A.2.
\end{proof}

Similar to the non-colluding case, the derived optimality condition \eqref{eq : objective rayleigh quotient coll} is cast as a class of NEPv \cite{cai:sjmaa:18} where the corresponding eigenvalue is equivalent to the objective function of \eqref{eq : objective maximize align coll}. Consequently, finding the principal eigenvector of ${\bF{B}}_{\sf c}^{-1}(\bar{\bF{f}}){\bF{A}_{\sf c}}(\bar{\bF{f}})$ is equivalent to finding the local optimal point. In the next subsection, we propose a novel method GPI-HIA for the colluding case (GPI-HIA (Coll)) which  finds the principal eigenvector efficiently.

\subsection{GPI-HIA for the Colluding Case}

The proposed GPI-HIA (Coll) iteratively updates the precoding vector as
\begin{equation}
\bar{\bF{f}}_{(t)} \leftarrow \frac{\bF{B}_{\sf c}^{-1}(\bar{\bF{f}}_{(t-1)})\bF{A}_{\sf c}(\bar{\bF{f}}_{(t-1)})\bar{\bF{f}}_{(t-1)}}{\left\|\bF{B}_{\sf c}^{-1}(\bar{\bF{f}}_{(t-1)})\bF{A}_{\sf c}(\bar{\bF{f}}_{(t-1)})\bar{\bF{f}}_{(t-1)} \right\|}.
\end{equation}
We repeat this process until the convergence criterion is met. For the convergence condition, we use $\left\|\bar{\bF{f}}_{(t)}-\bar{\bF{f}}_{(t-1)} \right\|<\epsilon$ where $\epsilon$ is a tolerance level. Algorithm \ref{alg:main c} describes the process.

\begin{algorithm} [t]
\caption{GPI-HIA (Coll)} \label{alg:main c} 

 {\bf{initialize}}: $\bar {\bf{f}}_{(0)}$ (MRT)\\
 Set the iteration count $t = 1$\\
\While {$\left\|\bar {\bf{f}}_{(t)} - \bar {\bf{f}}_{(t-1)} \right\| > \epsilon$}{
 Construct the matrix ${\bF{A}}_{\sf c}(\bar {\bF{f}}_{(t-1)})$  by using \eqref{eq : A coll}.\\
 Construct the matrix ${\bf{B}}_{\sf c} (\bar {\bF{f}}_{(t-1)})$ by using \eqref{eq : B coll}. \\
 Compute $\bar {\bF{f}}_{(t)} = {\bF{B}}
 _{\sf c}^{-1}(\bar{\bF{f}}_{(t-1)}){\bF{A}}_{\sf c}(\bar{\bF{f}}_{(t-1)})\bar{\bF{f}}_{(t-1)}$.\\
  {Normalize $\bar {\bf{f}}_{(t)} = \bar {\bf{f}}_{(t)}/\left\| \bar {\bf{f}}_{(t)}\right\|$}.\\
 $t \leftarrow t+1$.
 }
 \Return{\ }{${\bf \bar f} = {\bf \bar f}_{(t)}$}.\\
\end{algorithm}

\begin{remark} \normalfont (Fairness)
The fairness issue can be caused in the colluding case by an imbalance between the lower layer users' rates and the higher layer users' rates. We can address this issue by adopting the PF policy as in the non-colluding HIA case. 
\end{remark}

\section{Special Case: Precoding Optimization in Downlink NOMA Systems}
As a special case, our problems can reduce to a sum rate maximization problem in downlink MISO NOMA systems. In downlink MISO NOMA, each user has interference decoding capability in order to mitigate the amount of inter-user interference. 
Specifically, assuming that the BS serves $K$ users and the decoding order at each user is predetermined as $1 \rightarrow K$,  user $k$ first decodes $s_1, \cdots, s_{k-1}$ and removes the decoded messages, and decodes $s_k$ with the reduced amount of the interference. Without loss of generality, we let $|\bF{h}_1|^2< |\bF{h}_2|^2 <\cdots<|\bF{h}_K|^2$. 
To properly perform this SIC, it needs to be guaranteed that user $k$ successfully decodes $s_1, \cdots, s_{k-1}$, leading to that the information rate of $s_k$ is determined as $\bar R_k = \min\limits_{i \ge k} \{R_{i,k}\}$, where $R_{i,k}$ is the achievable rate of user $i$ for the message $s_k$ defined as \cite{zhu:tvt:20}
\begin{align}\label{eq : rate of noma}
    R_{i,k} = \log_2\left(1+\frac{\left| \bF{h}_{i}^{\sf H}\bF{f}_k\right|^2}{\sum_{j = k+1}^K\left|\bF{h}_{i}^{\sf H}\bF{f}_j\right|^2 + \frac{\sigma^2}{P}}\right).
\end{align}
In \eqref{eq : rate of noma}, $\bF{h}_i$ is the channel vector between the BS and user $i$.
Maximizing the sum rate in this system, an optimization problem with respect to the precoders is formulated as 
\begin{align}
	\max_{\bF{f}_1,\cdots,\bF{f}_K}&\sum\limits_{k=1}^K \bar R_k \label{eq : objective maximize align noma}
	\\ \label{eq:const noma}
 	\text{subject to }&\sum\limits_{k=1}^{K}||\bF{f}_{k}||^2  = 1.
\end{align}
Comparing \eqref{eq : objective maximize align non coll}, \eqref{eq : objective maximize align coll}, and \eqref{eq : objective maximize align noma}, we observe that the problem \eqref{eq : objective maximize align noma} is equivalent to the simplified form of our original setup, that can be reduced by ignoring the wiretapping lower priority users and assuming that only one user exists in each layer, i.e., $|\CMcal{L}_i| = 1$ for $i \in \{1,\cdots,K\}$. For this reason, the proposed GPI-HIA method is also applicable for  optimizing downlink MISO NOMA systems. 
In this section, extending our method further, we investigate how to apply the proposed method to solve \eqref{eq : objective maximize align noma} under the imperfect CSIT assumption.


\subsection{CSIT Acquisition and Relaxation}
On the contrary to the previous sections, we assume that the BS cannot have perfect knowledge of CSIT. 
Specifically, a limited feedback strategy is used to acquire the corresponding CSIT \cite{park:twc:16, park:wcl:16}, which renders the estimated CSIT $\hat {\bf{h}}_i$ as follows:
\begin{equation}\label{eq : estimated channel vector}
    \hat{\bF{h}}_{i} = \bF{U}_{i}\bF{\Lambda}_{i}^{\frac{1}{2}}\left( \sqrt{1-\kappa^2}\bF{g}_i + \kappa\bF{v}_i\right) = \bF{h}_i - \bF{e}_i,
\end{equation}
where ${\bf{e}}_i$ is the CSIT estimation error, $\bF{U}_i$ and $\bF{\Lambda}_i$ is the set of eigenvectors and eigenvalues of channel covariance matrix, $\bF{g}_i$ and $\bF{v}_i$ are drawn from IID $\CMcal{CN}(0,1)$. Here, the amount of feedback bits is implicitly controlled via the parameter $\kappa$; so that $\kappa$ decreases when the feedback bits increase, thereby the CSIT accuracy increases. In \eqref{eq : estimated channel vector}, the covariance of the error $\bF{e}_i$ is
\begin{equation}\label{eq : channel error cov}
    \bb{E}[\bF{e}_i\bF{e}_i^{\sf H}] = \bF{U}_{i}\bF{\Lambda}_{i}^{\frac{1}{2}}\left(2- 2\sqrt{1-\kappa^2}\right)\bF{\Lambda}_{i}^{\frac{1}{2}}\bF{U}_{i}^{\sf H} = \bF{\Phi}_i.
\end{equation}


With estimated channel vector, a lower bound on the achievable rate is obtained as
\begin{align}
    R_{i,k} &=\bb{E}\left[ \log_2\left(1+\frac{\left| {\bF{h}}_{i}^{\sf H}\bF{f}_k\right|^2}{\sum_{j = k+1}^K\left|{\bF{h}}_{i}^{\sf H}\bF{f}_j\right|^2  + \frac{\sigma^2}{P}}\right) \right] \nonumber \\
    &\overset{(a)}{\geq}\log_2\left(1+\frac{\bF{f}_k^{\sf H}\hat{\bF{h}}_i\hat{\bF{h}}_{i}^{\sf H}\bF{f}_k}{\sum_{j = k+1}^K \bF{f}_j^{\sf H}\hat{\bF{h}}_i\hat{\bF{h}}_{i}^{\sf H}\bF{f}_j +\sum_{j = k}^K\bF{f}_j^{\sf H}\bb{E}\left[\bF{e}_i\bF{e}_i^{\sf H}\right]\bF{f}_j + \frac{\sigma^2}{P}}\right) \nonumber \\
    &= \log_2\left(1+\frac{\bF{f}_k^{\sf H}\hat{\bF{h}}_i\hat{\bF{h}}_{i}^{\sf H}\bF{f}_k}{\sum_{j = k+1}^K \bF{f}_j^{\sf H}\hat{\bF{h}}_i\hat{\bF{h}}_{i}^{\sf H}\bF{f}_j +\sum_{j = k}^K\bF{f}_j^{\sf H}\Phi_{i}\bF{f}_j + \frac{\sigma^2}{P}}\right) = {R}^{\sf lb}_{i,k},\label{eq : lower bound of achievable rate}
\end{align}
where the expectation is for the randomness associated with the CSIT error, and $(a)$ comes from treating the CSIT estimation error as additive noise \cite{choi:twc:20, park:arxiv:21} and applying Jensen's inequality. With \eqref{eq : lower bound of achievable rate}, a lower bound on the information rate of $s_k$ is set as $\bar{R}_k^{\sf lb} = \min\limits_{i\geq k}\{ R_{i,k}^{\sf lb} \}$.
Accordingly, a relaxed sum rate maximization problem in NOMA system is reformulated as
\begin{align}
	\max_{\bF{f}_1,\cdots,\bF{f}_K}&\sum\limits_{k=1}^K \bar R_k^{\sf lb} \label{eq : new objective maximize align noma}
	\\ \label{eq:new const noma}
 	\text{subject to }&\sum\limits_{k=1}^{K}||\bF{f}_{k}||^2  = 1. 
\end{align}
Next, we put forth how to apply the proposed GPI-HIA method to solve \eqref{eq : new objective maximize align noma}. 

\subsection{Generalized Power Iteration for Downlink NOMA}
We follow the same step with the proposed GPI-HIA. We first rewrite \eqref{eq : lower bound of achievable rate} by representing the precoding vectors onto the higher dimensional space as in \eqref{eq : new f}:
\begin{align}
    R_{i,k}^{\sf lb} 
    =\log_2\left(\frac{\bar{\bF{f}}^{\sf H}\hat{\bF{A}}_{i,k}  \bar{\bF{f}}}{\bar{\bF{f}}^{\sf H}\hat{\bF{B}}_{i,k}  \bar{\bF{f}}} \right),\label{eq : rayleigh noma}
\end{align}
where 
\begin{align}
    \hat{\bF{A}}_{i,k} &= \text{blkdiag}\left(\bF{0}, \cdots,\bF{0},\underbrace{\hat{\bF{h}}_{i}\hat{\bF{h}}_{i}^{\sf H}+\bF{\Phi}_i}_{k\text{th block}},\cdots, \hat{\bF{h}}_{i}\hat{\bF{h}}_{i}^{\sf H}+\bF{\Phi}_i\right)
    +\frac{\sigma^2}{P}\bF{I}_{NK} \in \bb{C}^{NK\times NK}, \label{eq : effective channel matrix A noma} 
\end{align}
\begin{align}
    \hat{\bF{B}}_{i,k} &= \hat{\bF{A}}_{i,k} - \text{blkdiag}\left(\bF{0},\cdots,\bF{0},\underbrace{\hat{\bF{h}}_{i}\hat{\bF{h}}_{i}^{\sf H}}_{k\text{th block}},\bF{0},\cdots,\bF{0} \right)  \in\bb{C}^{NK\times NK}. \label{eq : effective channel matrix B noma}
\end{align}
Subsequently, we approximate the minimum function included in the optimization problem by leveraging \eqref{eq : rayleigh noma} and the LogSumExp technique:
\begin{align}
    \min_{i\geq k}\{R_{i,k}^{\sf lb} \} &\approx -\frac{1}{\alpha}\log\left(\sum_{i=k}^K e ^{-\alpha R_{i,k}^{\sf lb}}\right)
    =-\frac{1}{\alpha}\log\left(\sum_{i = k}^K\left(\frac{\bar{\bF{f}}^{\sf H}\hat{\bF{A}}_{i,k}\bar{\bF{f}}}{\bar{\bF{f}}^{\sf H}\hat{\bF{B}}_{i,k}\bar{\bF{f}}} \right)^{-\beta}\right),
\end{align}
where $\beta = \alpha\frac{1}{\log2}$. Finally, the problem \eqref{eq : new objective maximize align noma} is reformulated as
\begin{equation}\label{eq : rayleigh objective noma}
    \max_{\bar{\bF{f}}} \sum_{k = 1}^K-\frac{1}{\alpha}\log\left(\sum_{i = k}^K\left(\frac{\bar{\bF{f}}^{\sf H}\hat{\bF{A}}_{i,k}\bar{\bF{f}}}{\bar{\bF{f}}^{\sf H}\hat{\bF{B}}_{i,k}\bar{\bF{f}}} \right)^{-\beta}\right).
\end{equation}

To obtain a solution of \eqref{eq : rayleigh objective noma}, we drive a first-order KKT condition in the following corollary.
\begin{corollary}
The first-order KKT condition of the optimization problem \eqref{eq : rayleigh objective noma} is satisfied if the following holds.
\begin{equation}
    \bF{B}_{\sf NOMA}^{-1}(\bar{\bF{f}})\bF{A}_{\sf NOMA}(\bar{\bF{f}}) = \lambda_{\sf NOMA}(\bar{\bF{f}})\bar{\bF{f}} \label{eq : noma corollary}
\end{equation}
where
\begin{equation}\label{eq : A noma}
        {\bF{A}}_{\sf NOMA}(\bar{\bF{f}})=\lambda_{\sf NOMA}^{(\sf{num})}(\bar{\bF{f}})\times\sum_{k=1}^K\left({\sum\limits_{i=k}^{K}\left(\beta\left(\frac{\bar{\bF{f}}^{\sf H}\hat{\bF{A}}_{i,k}\bar{\bF{f}}}{\bar{\bF{f}}^{\sf H}\hat{\bF{B}}_{i,k}\bar{\bF{f}}}\right)^{-\beta}\frac{\hat{\bF{A}}_{i,k}}{\bar{\bF{f}}^{\sf H}\hat{\bF{A}}_{i,k}\bar{\bF{f}}} \right)}\Big/{\sum\limits_{i=k}^{K}\left(\frac{\bar{\bF{f}}^{\sf H}\hat{\bF{A}}_{i,k}\bar{\bF{f}}}{\bar{\bF{f}}^{\sf H}\hat{\bF{B}}_{i,k}\bar{\bF{f}}} \right)^{-\beta}} \right),
\end{equation}
\begin{equation}\label{eq : B noma}
        {\bF{B}}_{\sf NOMA}(\bar{\bF{f}}) =\lambda_{\sf NOMA}^{(\sf{den})}(\bar{\bF{f}})\times\sum_{k=1}^K\left({\sum\limits_{i=k}^{K}\left(\beta\left(\frac{\bar{\bF{f}}^{\sf H}\hat{\bF{A}}_{i,k}\bar{\bF{f}}}{\bar{\bF{f}}^{\sf H}\hat{\bF{B}}_{i,k}\bar{\bF{f}}}\right)^{-\beta}\frac{\hat{\bF{B}}_{i,k}}{\bar{\bF{f}}^{\sf H}\hat{\bF{B}}_{i,k}\bar{\bF{f}}} \right)}\Big/{\sum\limits_{i=k}^{K}\left(\frac{\bar{\bF{f}}^{\sf H}\hat{\bF{A}}_{i,k}\bar{\bF{f}}}{\bar{\bF{f}}^{\sf H}
        \hat{\bF{B}}_{i,k}\bar{\bF{f}}} \right)^{-\beta}} \right),
\end{equation}
with 
\begin{align}
    \lambda_{\sf NOMA} &= \sum_{k = 1}^K-\frac{1}{\alpha}\log\left(\sum_{i = k}^K\left(\frac{\bar{\bF{f}}^{\sf H}\hat{\bF{A}}_{i,k}\bar{\bF{f}}}{\bar{\bF{f}}^{\sf H}\hat{\bF{B}}_{i,k}\bar{\bF{f}}} \right)^{-\beta}\right) =\frac{\lambda_{\sf NOMA}^{(\sf num)}(\bar{\bF{f}})}{\lambda_{\sf NOMA}^{\sf (den)}(\bar{\bF{f}})}.
\end{align}
\begin{proof}
The proof is straightforward by removing the secrecy condition to the partial derivative of $\lambda_{\sf nc}(\bar{\bF{f}})$ in Appendix A.1.
\end{proof}
\end{corollary}
Now, to find the principal eigenvector of \eqref{eq : noma corollary}, we propose the GPI-NOMA algorithm. 
\begin{equation}
\bar{\bF{f}}_{(t)} \leftarrow \frac{\bF{B}_{\sf NOMA}^{-1}(\bar{\bF{f}}_{(t-1)})\bF{A}_{\sf NOMA}(\bar{\bF{f}}_{(t-1)})\bar{\bF{f}}_{(t-1)}}{\left\|\bF{B}_{\sf NOMA}^{-1}(\bar{\bF{f}}_{(t-1)})\bF{A}_{\sf NOMA}(\bar{\bF{f}}_{(t-1)})\bar{\bF{f}}_{(t-1)} \right\|}.
\end{equation}
We repeat this process until the convergence criterion is met. For the convergence condition, we use $\left\|\bar{\bF{f}}_{(t)}-\bar{\bF{f}}_{(t-1)} \right\|<\epsilon$, where $\epsilon$ is the predetermined tolerance level.

\section{Numerical Results}
In this section, we evaluate the ergodic sum secrecy spectral efficiency to validate the performance of the proposed GPI-HIA. For comparison, we consider the following baseline methods:
\begin{itemize}
    \item Maximum ratio transmission (MRT): since we assume a multigroup multicast scenario, it is not feasible to compute a precoding vector in a straightforward way. To obtain the precoding vectors, we construct $\bF{h}_k = \sum_{i\in\CMcal{L}_k}\bF{h}_{i,k}$ and use $\bF{h}_k$ as an effective channel vector for the users in $\CMcal{L}_k$. With this, 
    the precoding vectors are designed by matching the effective channel, i.e., $\bF{f}_k = \bF{h}_k$.
    \item Zero-forcing (ZF): we also use the effective channel vectors defined above. The precoding vectors are designed by following ZF methods as $\bF{F}  = \bF{H}\left(\bF{H}^{\sf H}\bF{H} \right)^{-1}$, where $\bF{F} = \left[\bF{f}_1,\cdots,\bF{f}_K \right]$ and $\bF{H} = \left[\bF{h}_1,\cdots,\bF{h}_K \right]$.
    \item Multicast weighted minimum mean square error (WMMSE): in this method, the conventional WMMSE precoding \cite{christensen:twc:08} is modified for a multi-group multicast scenario \cite{yalcin:tcom:19}. Specifically, using the correspondence between mutual information and MMSE, a WMSE minimization problem was formulated. To release the non-convexity of the problem, the original problem was reformulated to a smooth constrained optimization problem with auxiliary variables. To solve this reformulated problem, the exponential penalty method was adopted. We note that this method finds a local optimal point of a general multi-group multicast scenario, while it does not incorporate our HIA condition.
\end{itemize}
In the proposed GPI-HIA setting, the LogSumExp approximation parameter $\alpha$ is determined as follows: we initialize it as $\alpha_1 = 10$ and if the proposed algorithm does not converge within a certain number of iterations $T = 50$, updating it as $\alpha_{n+1} = 0.9\alpha_n$.

\begin{figure}[!t]
     \centerline{\resizebox{0.5\columnwidth}{!}{\includegraphics{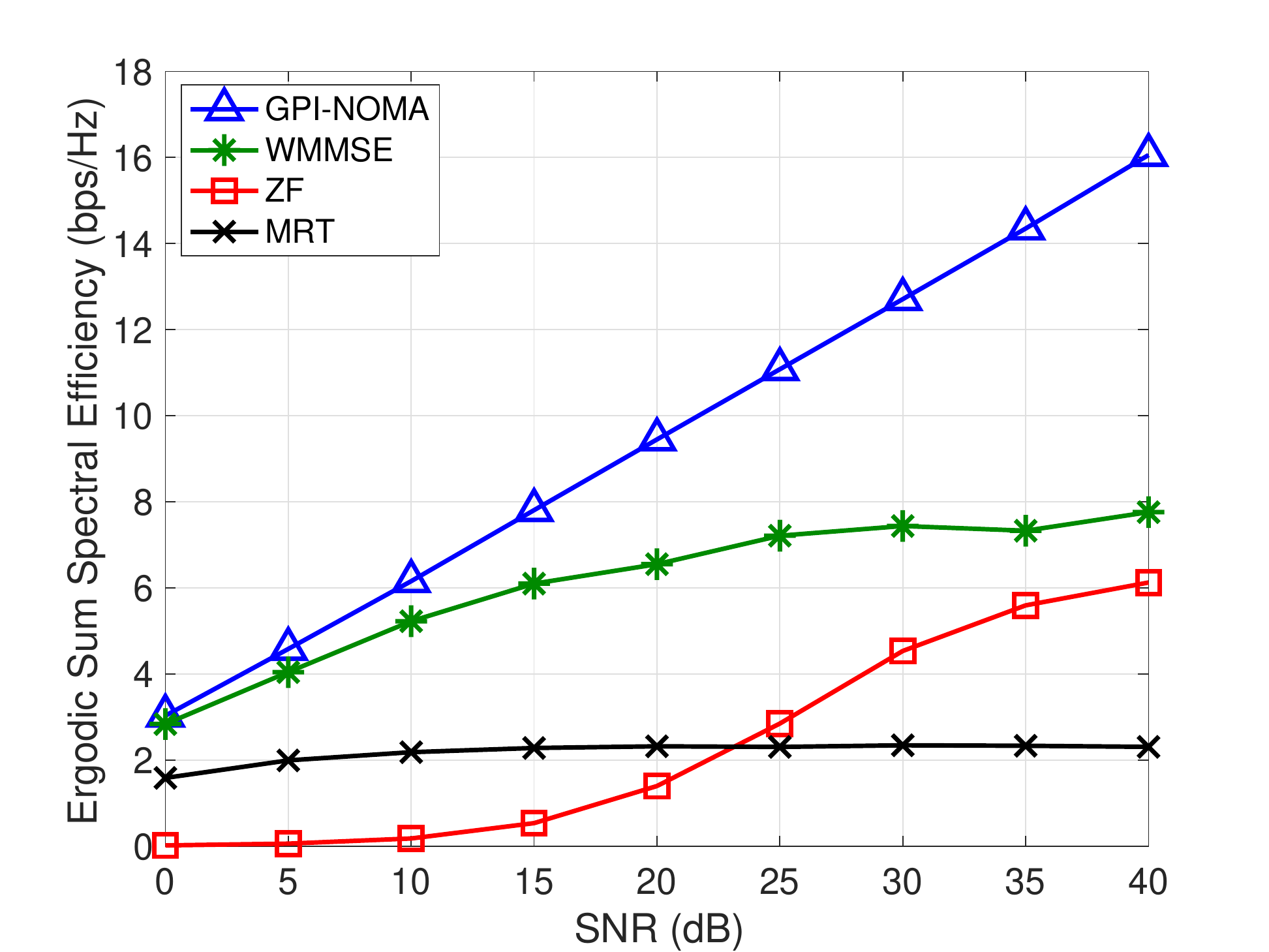}}}
     \caption{Ergodic sum spectral efficiency in downlink NOMA system when SNR is increasing. The simulation parameters are $N = 4$, $K = 8$, $\kappa = 0.4$ $|\CMcal{L}_k| = 1$, $\epsilon = 0.01$, $\Delta_k  = \pi/6$ for $k\in \{1,\cdots,K\}$, and $\theta_k = \frac{\pi}{6}$.}
     \label{fig:result NOMA}
\end{figure}

 
\subsection{Downlink MISO NOMA}
At first, we illustrate the comparison of the ergodic sum spectral efficiency for the downlink MU-MISO system in \fig{fig:result NOMA} when the signal-to-noise ratio (SNR) is increasing. As baseline methods, we use the conventional precoding methods: 1) MRT, 2) ZF, and 3) WMMSE. 
Note that the WMMSE method in the downlink NOMA case is a reduced version of the multicast WMMSE \cite{yalcin:tcom:19} by assuming that each layer has only one user. 
We assume that all the users are clustered in a specific location. 

We observe that the proposed algorithm outperforms the baseline methods in all the SNR regimes. Specifically, when the SNR is $40$ dB, it shows about $200\%$ gains over the WMMSE precoding method. This is because the proposed algorithm is designed by incorporating imperfect CSIT.

\begin{figure}[!t]
     \centerline{\resizebox{0.6\columnwidth}{!}{\includegraphics{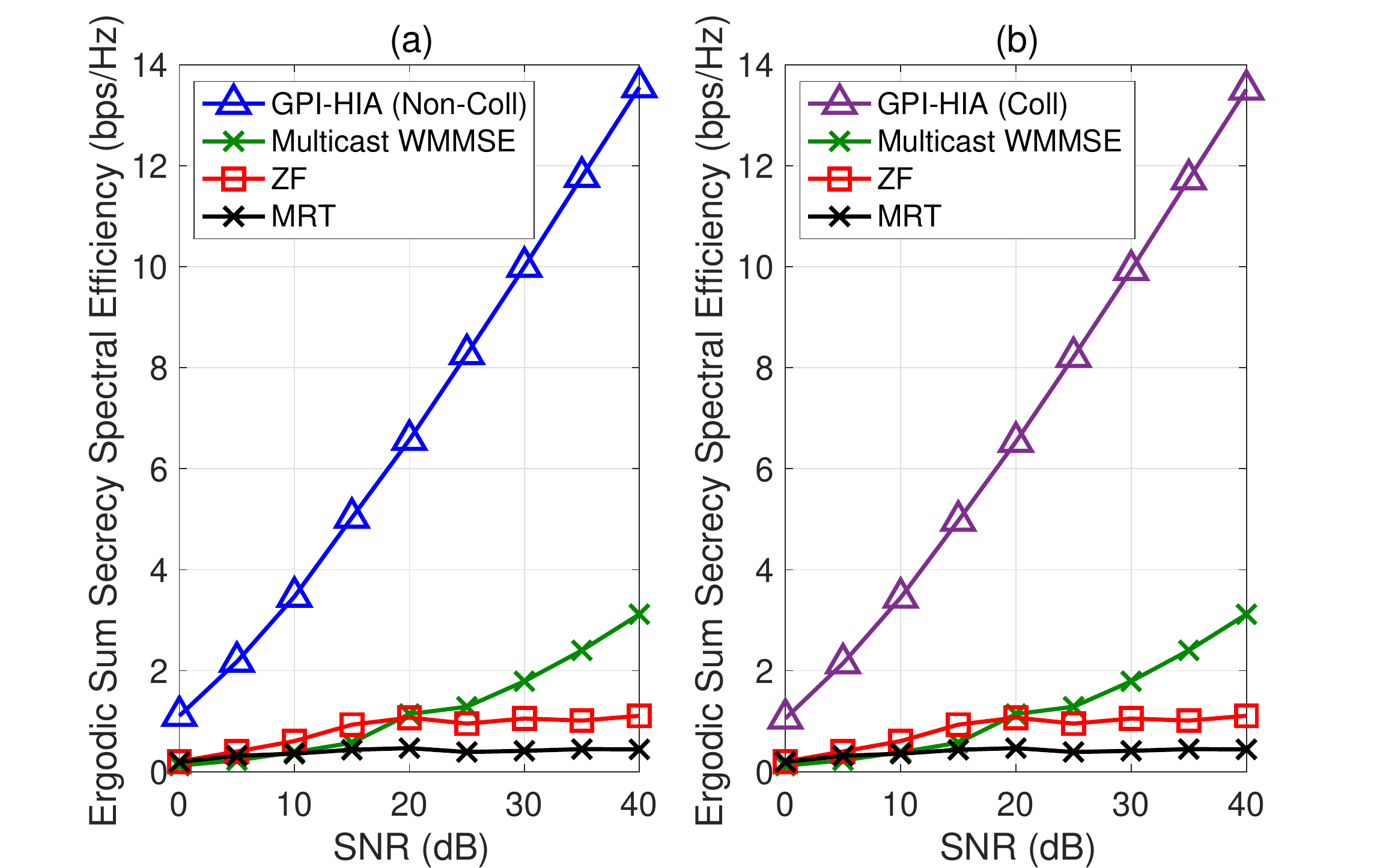}}}
     \caption{Ergodic sum secrecy rate with (a) non-colluding case and (b) colluding case when SNR is increasing. The simulation parameters are $N = 6$, $K = 3$, $|\CMcal{L}_k| = 2$, $\epsilon = 0.01$, $\beta_m = 1$, $\Delta_m  = \pi/6$ for $m\in \CMcal{M}$, and $\theta_m = [0,2\pi]$.}
     \label{fig:result per SNR}
\end{figure}

\subsection{Secrecy Rate Comparison per SNR} 
In this subsection, 
we illustrate the comparison for the ergodic sum secrecy spectral efficiency per SNR in \fig{fig:result per SNR}. 
\fig{fig:result per SNR}(a) shows the performance comparison in the non-colluding case and \fig{fig:result per SNR}(b) shows the performance comparison in the colluding case.
In the both case, the proposed algorithm provides the highest performance compared to the baseline method. The noticeable point is that more gains are achieved in the high SNR regime. Especially, when the transmit SNR is 40dB, the proposed GPI-HIA provides about $460\%$ gains over WMMSE. This is because the GPI-HIA is designed by properly incorporating the intertwined HIA condition. If the HIA condition is not reflected into the algorithm, the degrees-of-freedom of the sum secrecy spectral efficiency is degraded as observed in \fig{fig:result per SNR}.
\subsection{Secrecy Rate Comparison per the Number of the Users}
In this subsection, we illustrate the simulation results in \fig{fig:result per M} when the number of users increases. 
\fig{fig:result per M}(a) shows the performance comparison in the non-colluding case and \fig{fig:result per M}(b) shows the performance comparison in the colluding case. 
Since a multi-group multicast message scenario is considered, it is natural that the sum secrecy spectral efficiency decreases as the number of users increases. In \fig{fig:result per M}, ZF and MRT achieve almost 0 sum secrecy spectral efficiency if the number of users is larger than $15$. On contrary to that, the proposed GPI-HIA gets robust sum secrecy spectral efficiency compared to baseline methods. 




\begin{figure}[!t]
     \centerline{\resizebox{0.6\columnwidth}{!}{\includegraphics{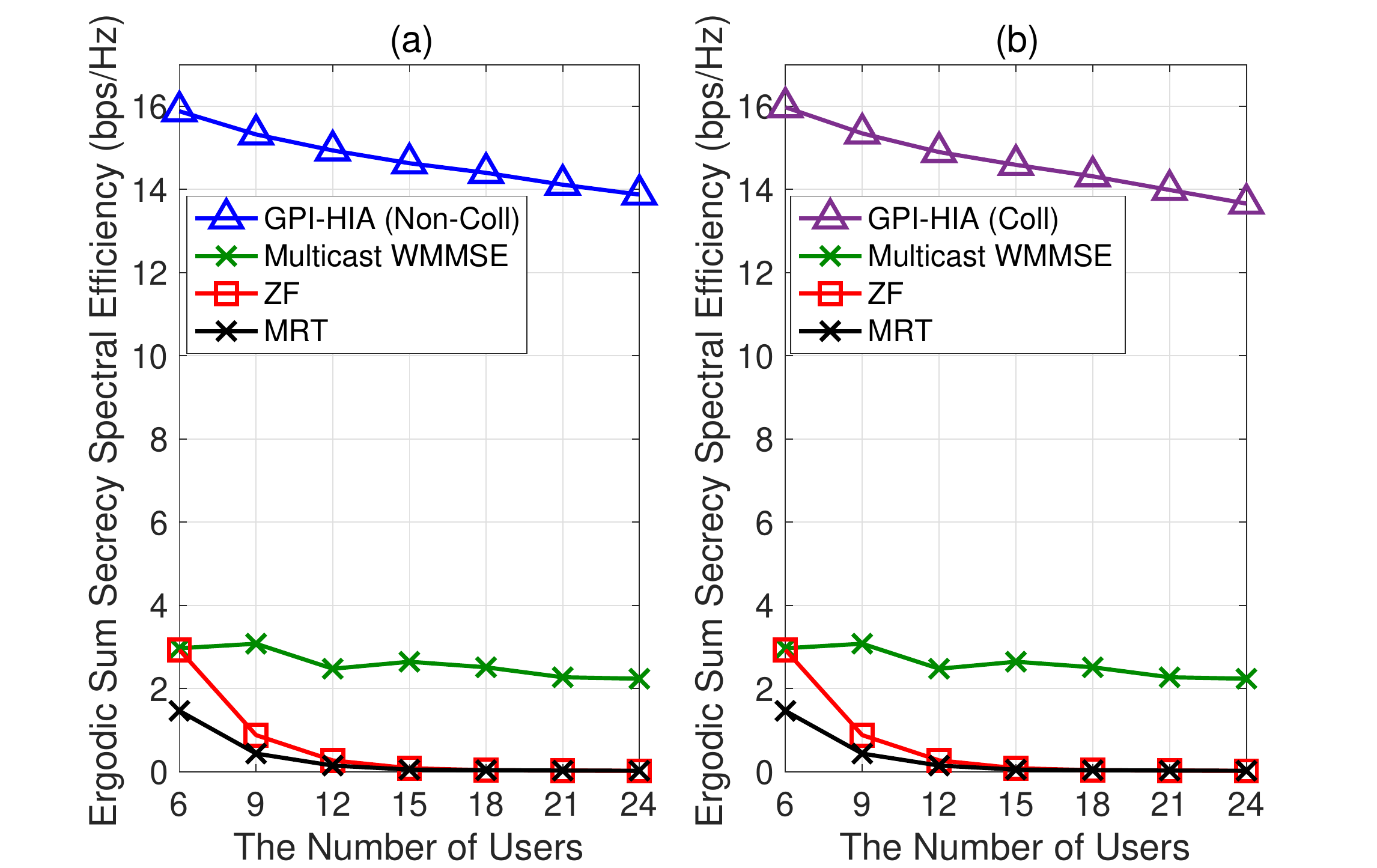}}}
     \caption{Ergodic sum secrecy rate when the number of users is increasing in (a) non-colluding case and (b) colluding case. The simulation parameters are $N = 24$, $K = 3$ $|\CMcal{L}_k| = \frac{|\CMcal{M}|}{K}$, $\epsilon = 0.01$, SNR = $40$dB, $\beta_m=1$, $\Delta_m = \pi/6$ for $m\in \CMcal{M}$, and $\theta_m = [0,2\pi]$.}
     \label{fig:result per M}
\end{figure}

\subsection{Convergence}
\fig{fig : convergence} shows the convergence results in terms of $\lambda_{\sf nc}(\bar{\bF{f}})$, $\lambda_{\sf c}(\bar{\bF{f}})$ and residual $\left\|\bar{\bF{f}}_{(t)}-\bar{\bF{f}}_{(t-1)} \right\|$ during 50 iterations. It is observed that both GPI-HIA (Non-Coll) and GPI-HIA (Coll) converge within  10 iterations. 
This provides an empirical evidence that the proposed GPI-HIA reaches a desirable local optimal point fast. 
\begin{figure}[!t]
     \centerline{\resizebox{0.7\columnwidth}{!}{\includegraphics{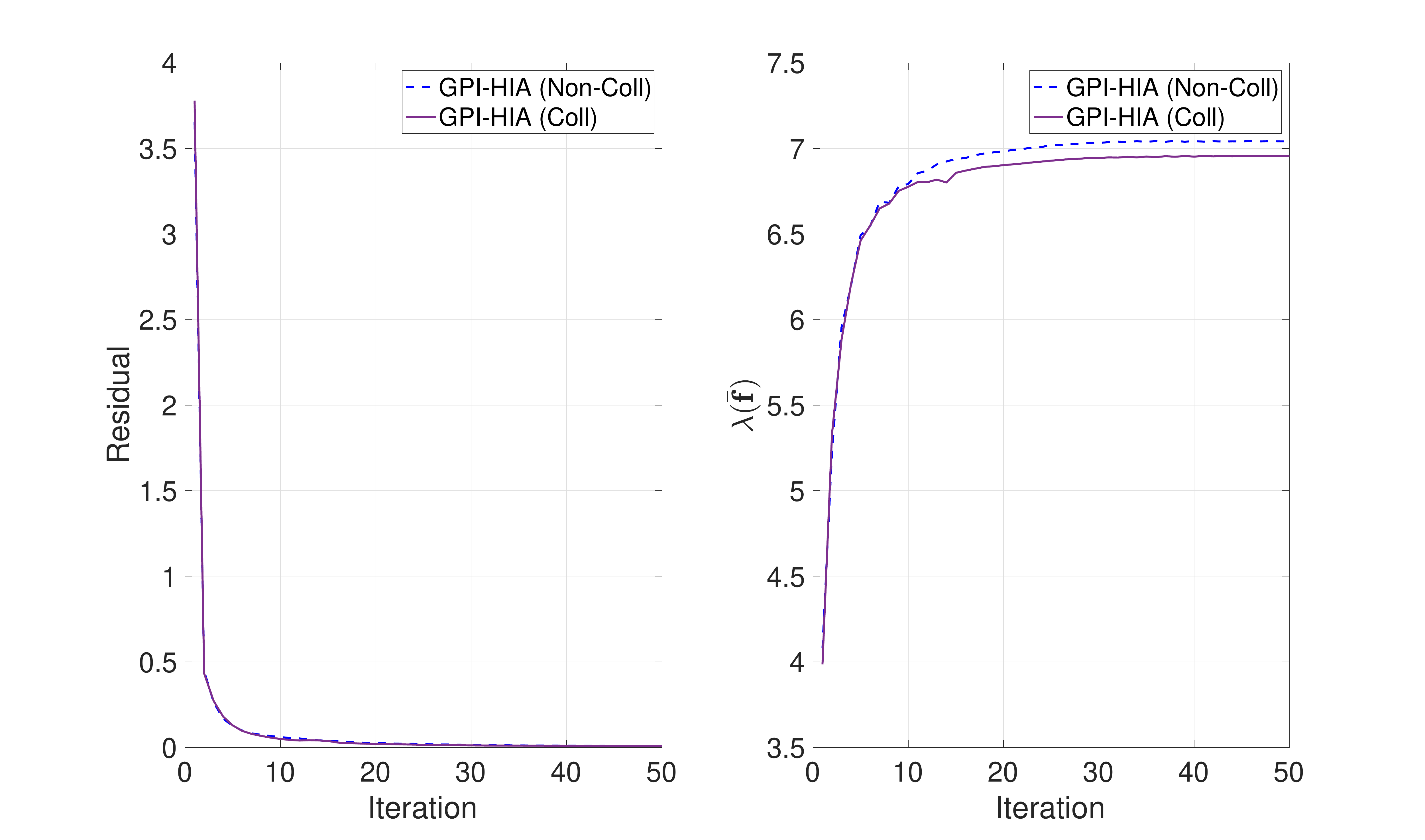}}}
     \caption{Convergence results in terms of $\lambda(\bar{\bF{f}})$ and the residual $\left\|\bar{\bF{f}}_{(t)}-\bar{\bF{f}}_{(t-1)} \right\|$. The simulation parameters are $N = 6$, $K = 3$, $|\CMcal{L}_1| = 3$, $|\CMcal{L}_2| = 2$, $|\CMcal{L}_3| = 1$, $\epsilon = 0.01$, SNR = $20$dB, $\Delta_m = \pi/6$ for $m\in \CMcal{M}$, $\theta_m = [0,2\pi]$ and $\beta_m = 1$.}
     \label{fig : convergence}
\end{figure}

\subsection{Fairness}
In this subsection, we draw system level simulation results for resolving the fairness issue by incorporating the PF policy in the GPI-HIA. 
For a large scale fading parameter $\beta_m$ used in system level simulation, we adopt the log-distance pathloss model in \cite{erceg:jsac:99}. 
The distance between the BS and the users is 100m to 500m. We consider a 2.4 GHz carrier frequency with 10MHz bandwidth, -174 dBm/Hz noise power spectral density, and 9 dB noise figure. 

The empirical cumulative distribution function (CDF) for the secrecy rate is depicted in \fig{fig : result cdf}(a). The GPI-HIA-PF shows a steeper curve than the GPI-HIA, i.e., the proposed algorithm with PF policy provides more uniform secrecy rate to the users by resolving the fairness issue. 
For better understanding, we also present a snapshot of the achieved secrecy spectral efficiency for each layer in \fig{fig : result cdf}(b) and \fig{fig : result cdf}(c).


\subsubsection{Non-Colluding Case}
as the blue bar in the \fig{fig : result cdf}(b) shows, the power is concentrated on the message $s_1$ intended to the lowest layer users. After applying the PF policy to the GPI-HIA, the yellow bar in the \fig{fig : result cdf}(b), the secrecy rate for the messages intended to the higher layer increases.


\subsubsection{Colluding Case}
as the purple bar in \fig{fig : result cdf}(c) shows, we confirm that the power is concentrated on the message intended to the lower layer users. After applying the PF policy to the GPI-HIA, the green bar in \fig{fig : result cdf}(c), the secrecy rate for the messages intended to the higher layer increases.

\begin{figure}[!t]
     \centerline{\resizebox{0.7\columnwidth}{!}{\includegraphics{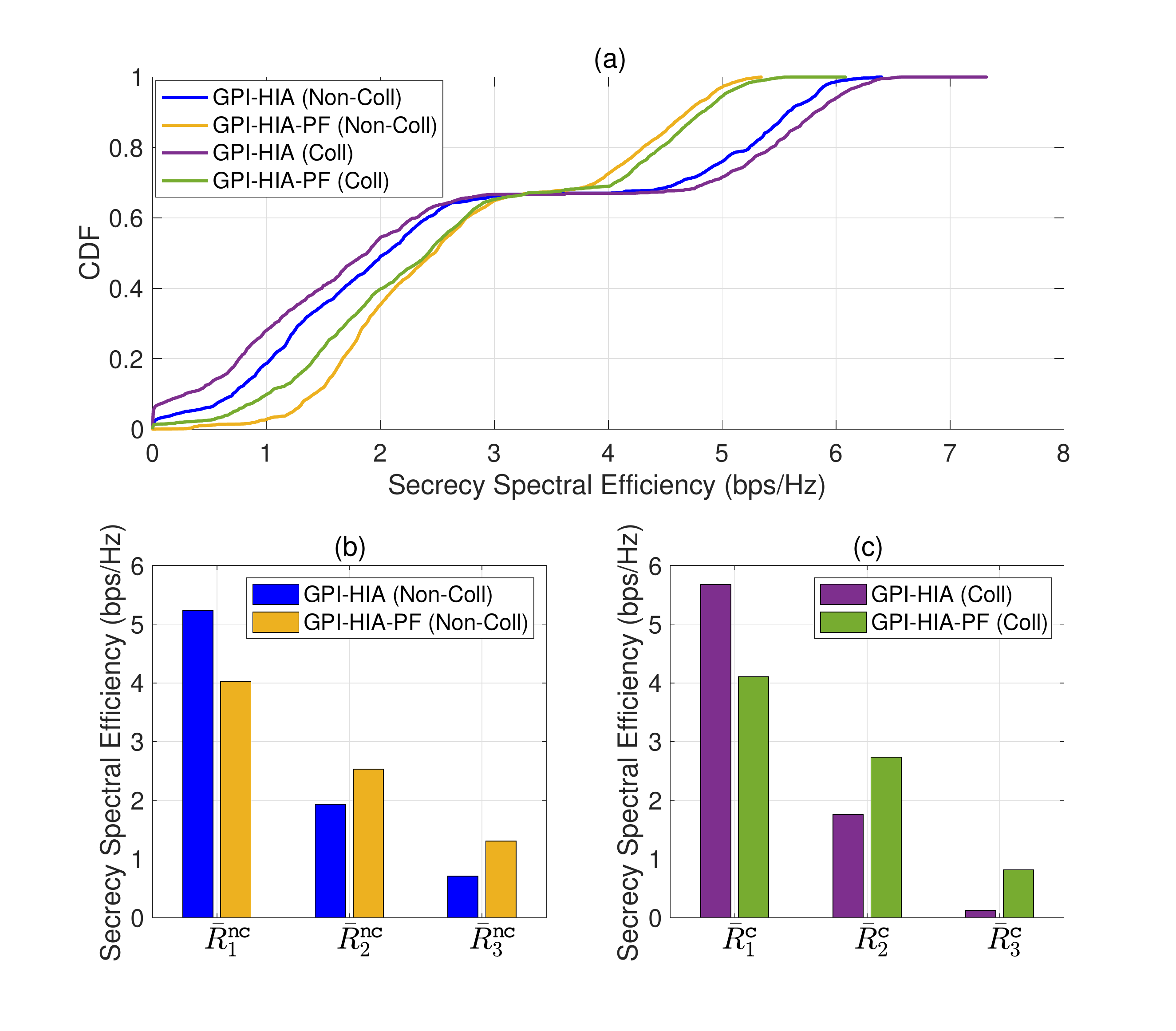}}}
     \caption{Secrecy rate of each message in (a) non-colluding case and (b) colluding case. The simulation parameters are $N = 6$, $K = 3$, $|\CMcal{L}_1| = 3$, $|\CMcal{L}_2| = 2$, $|\CMcal{L}_3| = 1$, $\epsilon = 0.01$, $\delta = 0.2$, SNR = $20$dB, $\Delta_m = \pi/6$ for $m\in \CMcal{M}$ and $\theta_m = [0,2\pi]$.}
     \label{fig : result cdf}
\end{figure}

\section{Conclusion}
In this paper, we have considered a new security model that generalizes conventional physical layer security, referred to as HIA. Although the HIA is useful to reflect a hierarchical security structure, optimizing such a system is highly challenging due to its intertwined rate formations. Resolving the challenges, we have proposed new precoding methods to maximize the sum secrecy rate by considering the non-colluding and the colluding cases. Specifically, we have approximated non-smooth functions by using the LogSumExp technique and have reformulated the optimization problem as a form of Rayleigh quotients with a higher dimensional vector. With this form, we have derived the first-order optimality condition and cast this condition into a NEPv in which our objective function and precoding vector are mapped to the eigenvalue and eigenvector, respectively. Accordingly, our proposed methods compute the principal eigenvector to obtain a local optimal solution. As a byproduct, we have shown that the proposed precoding framework is also applicable to a sum rate maximization algorithm for downlink NOMA systems because the presented HIA model includes downlink NOMA with a fixed decoding order as a special case. 
Simulation results have validated that our methods provide significant performance improvement over existing baseline methods. In addition to this, our methods are beneficial in an implementation perspective since no off-the-shelf solver is used. 


\section*{Appendix A.1\\Proof of Lemma 1}
We first derive first-order KKT condition, i.e., $\frac{\partial\lambda_{\sf nc}(\bar{\bF f})}{\partial\bar{\bF f}^{\sf H}} = 0$ of the problem \eqref{eq : objective rayleigh quotient non coll} where 
\begin{equation}
    \lambda_{\sf nc}(\bar{\bF f})=\sum_{k=1}^K\left(-\frac{1}{\alpha}\log\left(\sum_{\ell = k}^K\sum_{i\in \CMcal{L}_{\ell}}\left(\frac{\bar{\bF{f}}^{\sf H}\bF{A}_{k,i,\ell}\bar{\bF{f}}}{\bar{\bF{f}}^{\sf H}\bF{B}_{k,i,\ell}\bar{\bF{f}}}\right)^{-\beta} \right)
     -\frac{1}{\alpha}\log\left(\sum_{\ell'=1}^{k-1}\sum_{i'\in\CMcal{L}_{\ell'}}\left(\frac{\bar{\bF{f}}^{\sf H}\bF{A}_{k,i',\ell'}\bar{\bF{f}}}{\bar{\bF{f}}^{\sf H}\bF{B}_{k,i',\ell'}\bar{\bF{f}}}\right)^{\beta}\right)\right).
\end{equation}
Using the derivative property,
the partial derivative of $\lambda_{\sf nc}(\bar{\bF f})$ is obtained as
\begin{multline}\label{eq : partial derivative of lambda nc}
    \frac{\lambda_{\sf nc}(\bar{\bF{f}})}{\partial\bar{\bF{f}}^{\sf H}}=
    -\frac{1}{\alpha}\sum_{k=1}^K\left(\frac{\sum\limits_{\ell=k}^K\sum\limits_{i\in\CMcal{L}_{\ell}}\left(-\beta\left(\frac{\bar{\bF{f}}^{\sf H}\bF{A}_{k,i,\ell}\bar{\bF{f}}}{\bar{\bF{f}}^{\sf H}\bF{B}_{k,i,\ell}\bar{\bF{f}}}\right)^{-\beta}\left(\frac{\bF{A}_{k,i,\ell}\bar{\bF{f}}}{\bar{\bF{f}}^{\sf H}\bF{A}_{k,i,\ell}\bar{\bF{f}}}
    - \frac{\bF{B}_{k,i,\ell}\bar{\bF{f}}}{\bar{\bF{f}}^{\sf H}\bF{B}_{k,i,\ell}\bar{\bF{f}}}  \right) \right)}{\sum\limits_{\ell=k}^K\sum\limits_{i\in\CMcal{L}_{\ell}}\left(\frac{\bar{\bF{f}}^{\sf H}\bF{A}_{a,n}\bar{\bF{f}}}{\bar{\bF{f}}^{\sf H}\bF{B}_{a,n}\bar{\bF{f}}}\right)^{-\beta}}\right.\\
\left.-\frac{\sum\limits_{\ell'=1}^{k-1}\sum\limits_{i'\in\CMcal{L}_{\ell'}}\left(\beta\left(\frac{\bar{\bF{f}}^{\sf H}\bF{A}_{k,i',\ell'}\bar{\bF{f}}}{\bar{\bF{f}}^{\sf H}\bF{B}_{k,i',\ell'}\bar{\bF{f}}}\right)^{\beta}\left(\frac{\bF{A}_{k,i',\ell'}\bar{\bF{f}}}{\bar{\bF{f}}^{\sf H}\bF{A}_{k,i,\ell'}\bar{\bF{f}}}- \frac{\bF{B}_{k,i',\ell'}\bar{\bF{f}}}{\bar{\bF{f}}^{\sf H}\bF{B}_{k,i',\ell'}\bar{\bF{f}}}  \right) \right)}{\sum\limits_{\ell'=1}^{k-1}\sum\limits_{i'\in\CMcal{L}_{\ell'}}\left(\frac{\bar{\bF{f}}^{\sf H}\bF{A}_{k,i',\ell'}\bar{\bF{f}}}{\bar{\bF{f}}^{\sf H}\bF{B}_{k,i',\ell'}\bar{\bF{f}}}\right)^{\beta}}\right).
\end{multline}
The first-order KKT condition holds when \eqref{eq : partial derivative of lambda nc} equals to 0. Defining $\bF{A}_{\sf nc}(\bar{\bF{f}})$, $\bF{B}_{\sf nc}(\bar{\bF{f}})$ and $\lambda_{\sf nc}(\bar{\bF f})$ as \eqref{eq : A non coll}, \eqref{eq : B non coll} and \eqref{eq : lambda non coll}, the first-order KKT condition is rearranged as
\begin{equation}
    {\bF{A}}_{\sf nc}(\bar{\bF{f}})\bar{\bF{f}} = \lambda_{\sf nc}(\bar{\bF{f}}){\bF{B}}_{\sf nc}(\bar{\bF{f}})\bar{\bF{f}}\Leftrightarrow \bF{B}_{\sf nc}^{-1}(\bar{\bF{f}}){\bF{A}}_{\sf nc}(\bar{\bF{f}})\bar{\bF{f}} = \lambda_{\sf nc}(\bar{\bF{f}})\bar{\bF{f}}.
\end{equation}
This completes the proof. \qed

\section*{Appendix A.2\\Proof of Lemma 1}
We first derive first-order KKT condition, i.e., $\frac{\partial\lambda_{\sf c}(\bar{\bF f})}{\partial\bar{\bF f}^{\sf H}} = 0$ of the problem \eqref{eq : objective rayleigh quotient coll} where 
\begin{equation}
    \lambda_{\sf c}(\bar{\bF{f}}) = \sum_{k=1}^K\left( -\frac{1}{\alpha}\log\left(\sum_{\ell=k}^K\sum_{i\in\CMcal{L}_\ell}\left(\frac{\bar{\bF{f}}^{\sf H}\bF{A}_{k,i,\ell}\bar{\bF{f}}}{\bar{\bF{f}}^{\sf H}\bF{B}_{k,i,\ell}\bar{\bF{f}}} \right)^{-\beta}\right)-\frac{1}{\log 2}\log\left(\sum_{\ell' = 1}^{k-1}\sum_{i'\in\CMcal{L}_{\ell'}}\left(\frac{\bar{\bF{f}}^{\sf H}\bF{C}_{k,i',\ell'}\bar{\bF{f}}}{\bar{\bF{f}}^{\sf H}\bF{D}_{k,i',\ell'}\bar{\bF{f}}} \right) \right) \right).
\end{equation}
Using the similar technique in Lemma \ref{lem:noncoll}, the partial derivative of $\lambda_{\sf c}(\bar{\bF f})$ is obtained by using the above calculation as 
\begin{multline}\label{eq : partial derivative lambda c}
            \frac{\partial \lambda_{\sf c}(\bar{\bF{f}})}{\partial\bar{\bF{f}}^{\sf H}} = \sum_{k=1}^K\left(-\frac{1}{\alpha}\frac{\sum\limits_{\ell = k}^K\sum\limits_{i\in\Cal{L}_{\ell}}\left(-\beta\left(\frac{\bar{\bF{f}}^{\sf H}\bF{A}_{k,i,\ell}\bar{\bF{f}}}{\bar{\bF{f}}^{\sf H}\bF{B}_{k,i,\ell}\bar{\bF{f}}}\right)^{-\beta}\left(\frac{\bF{A}_{k,i,\ell}\bar{\bF{f}}}{\bar{\bF{f}}^{\sf H}\bF{A}_{k,i,\ell}\bar{\bF{f}}}-\frac{\bF{B}_{k,i,\ell}\bar{\bF{f}}}{\bar{\bF{f}}^{\sf H}\bF{B}_{k,i,\ell}\bar{\bF{f}}} \right)  \right)}{\sum\limits_{\ell = k}^K\sum\limits_{i\in\Cal{L}_{\ell}}\left(\frac{\bar{\bF{f}}^{\sf H}\bF{A}_{k,i,\ell}\bar{\bF{f}}}{\bar{\bF{f}}^{\sf H}\bF{B}_{k,i,\ell}\bar{\bF{f}}}\right)^{-\beta}}\right.\\
            \left. -\frac{1}{\log2}\frac{\sum\limits_{\ell' = 1}^{k-1}\sum\limits_{i'\in\Cal{L}_{\ell'}}\frac{\bar{\bF{f}}^{\sf H}\bF{C}_{k,i',\ell'}\bar{\bF{f}}}{\bar{\bF{f}}^{\sf H}\bF{D}_{k,i',\ell'}\bar{\bF{f}}}\left(\frac{\bF{C}_{k,i',\ell'}\bar{\bF{f}}}{\bar{\bF{f}}^{\sf H}\bF{C}_{k,i',\ell'}\bar{\bF{f}}}-\frac{\bF{D}_{k,i',\ell'}\bar{\bF{f}}}{\bar{\bF{f}}^{\sf H}\bF{D}_{k,i',\ell'}\bar{\bF{f}}} \right)}{\sum\limits_{\ell' = 1}^{k-1}\sum\limits_{i'\in\Cal{L}_{\ell'}}\frac{\bar{\bF{f}}^{\sf H}\bF{C}_{k,i',\ell'}\bar{\bF{f}}}{\bar{\bF{f}}^{\sf H}\bF{D}_{k,i',\ell'}\bar{\bF{f}}}} \right).
\end{multline}
The first-order KKT condition holds when \eqref{eq : partial derivative lambda c} equals to 0. Defining $\bF{A}_{\sf c}(\bar{\bF{f}})$, $\bF{B}_{\sf c}(\bar{\bF{f}})$ and $\lambda_{\sf c}(\bar{\bF f})$ as \eqref{eq : A coll}, \eqref{eq : B coll} and \eqref{eq : lambda coll}, the first-order KKT condition is rearranged as
\begin{equation}
    {\bF{A}}_{\sf c}(\bar{\bF{f}})\bar{\bF{f}} = \lambda_{\sf c}(\bar{\bF{f}}){\bF{B}}_{\sf c}(\bar{\bF{f}})\bar{\bF{f}}\Leftrightarrow \bF{B}_{\sf c}^{-1}(\bar{\bF{f}}){\bF{A}}_{\sf c}(\bar{\bF{f}})\bar{\bF{f}} = \lambda_{\sf c}(\bar{\bF{f}})\bar{\bF{f}}.
\end{equation}

This completes the proof. \qed

\bibliographystyle{IEEEtran}
\bibliography{ref_HIA}

\end{document}